%% file: main.tex
\documentclass{scrartcl}

\usepackage{graphicx} 
\usepackage[utf8x]{inputenc}
\usepackage{amssymb, amsmath}
\usepackage{geometry}
\usepackage{amsthm}
\newtheorem{definition}{Definition}[section]
\newtheorem{theorem}[definition]{Theorem}

\newtheorem{lemma}[definition]{Lemma}
\newtheorem{proposition}[definition]{Proposition}
\usepackage{nicefrac}
\usepackage{url}
\title{$C^0$-inextendibility of the Kasner spacetime}
\author{Benedikt Miethke}
\date{August 2024}

\begin{document}
\maketitle

\begin{abstract}
    The Kasner spacetime is a cosmological model of an anisotropic expanding universe without matter and is an exact solution of the Einstein vacuum equations $Ric(g)=0$. It depends on a choice of so-called Kasner exponents and if one of these is negative, then the Kretschmann scalar blows up as $t\rightarrow 0$, i.e. there exists a curvature singularity. Thus, it is manifestly inextendible as a Lorentzian manifold with a twice differentiable metric. In this Master's thesis \footnote{Master's thesis written at Universität Wien, supervised by Jan Sbierski.\\ Original document: \url{https://utheses.univie.ac.at/detail/72128/}} we proof that it is even inextendible as a Lorentzian manifold with merely continuous metric, which is a stronger statement. We do so by adapting the proof of the $C^0$-inextendibility of the maximal analytically extended Schwarzschild spacetime established by Jan Sbierski in \cite{Schwarzschild2018}.
\end{abstract}

\tableofcontents

\pagestyle{headings}
\pagenumbering{arabic}

\include{Introduction}

\include{Theory}
\include{Kasner}

\bibliographystyle{plain}
\bibliography{main}

\end{document}

%% file: Introduction.tex
\section{Introduction}
\subsection{Motivation}
The question whether a given solution to the Einstein equations is maximal or can be extended as a (weak) solution motivates the study of low regularity (in)extendibility of a Lorentzian manifold. Of particular interest is the so-called \textit{strong cosmic censorship conjecture}, originally proposed by Roger Penrose. It states that, generally, the theory of general relativity is deterministic, i.e. we should be able predict the fate of all classical observers. One possible mathematical formulation of this conjecture is the following (\cite{Schwarzschild2016}):\\
\begin{addmargin}[25pt]{5pt}
For generic asymptotically flat initial data for the vacuum Einstein equations\\
$Ric(g)=0$, the maximal globally hyperbolic development is inextendible as a\\
suitably regular Lorentzian manifold.
\end{addmargin}
\bigskip

The meaning of \textquotedblleft suitably regular\textquotedblright and \textquotedblleft generic\textquotedblright initial data is still debated to this day. However it is clear, that a better understanding of $C^0$-inextendibility will also yield inextendibility results in all other regularity classes. So the study of $C^0$-extensions can give useful insights for research on the strong cosmic censorship conjecture, even if the true \textquotedblleft suitable regularity\textquotedblright, for which the conjecture should be stated, might be of higher regularity than $C^0$.\\
\\
In this Master's thesis we will prove the $C^0$-inextdendibility of the Kasner spacetime, which was (to the best of the authors knowledge) first suspected in \cite{Schwarzschild2016}. The proof will follow the strategy established in \cite{Schwarzschild2018} and thus confirm that this strategy of proving $C^0$-inextendibility of the maximal analytically extended Schwarzschild spacetime can be generalized to other spacetimes with suitable properties. In the following we will always refer to the maximal analytically extended Schwarzschild spacetime simply as the Schwarzschild spacetime.

\subsection{Comparison of Schwarzschild and Kasner proof}
The Schwarzschild spacetime and Kasner spacetime are fundamentally different as the first models a static, non-rotating, spherical symmetric black hole and the latter is a cosmological model of an anisotropic expanding/contracting universe with a big bang/big crunsh. However, as we will see, they share important properties that allow us to prove $C^0$-inextendibility with the same strategy.\\
\\
Let us first recall the proof idea for the Schwarzschild spacetime in \cite{Schwarzschild2018}:\\
We argue by contradiction and assume there exists a $C^0$-extension of the Schwarzschild spacetime. Since the Schwarzschild spacetime is globally hyperbolic, a result established in \cite{GallowayLingSbierski} states that then, there exists a (without loss of generality) future-directed timelike geodesic leaving the Schwarzschild spacetime. It follows that this geodesic either "leaves through the exterior", i.e. timelike/null infinity $r\rightarrow\infty$, or it "leaves through the interior", i.e. the curvature singularity $r=0$.\\
The first case is easily ruled out, since such geodesics are future complete in the Schwarzschild spacetime (compare with Theorem 1 in \cite{GallowayLingSbierski}).\\
So we are left with the case that the geodesic leaves through the curvature singularity $r=0$. In this case, we can find a chart around the endpoint of the geodesic at $r=0$, where we can map the future boundary of the extension as an achronal Lipschitz graph and we have good control of the metric (i.e. it is close to the Minkowski metric). Furthermore, we can consider the future directed timelike curve given by the $x^0$-coordinate in the chart, which leaves the spacetime through the same point as the geodesic. It turns out that there exists a point on this curve, close to the endpoint on the boundary, such that its whole future in the Schwarzschild spacetime is mapped into the chart below the graph and that we can reach this point from the past with a timelike curve from any space-coordinate of the chart (far enough in the past in the chart).
This is the most difficult step of the proof, since we somehow need to control the causality of the spacetime, to map it below the graph. We do this by using the isometries we have because of the spherical symmetry and the coordinate vector field $\partial_t$ being stationary and the fact, that the Schwarzschild spacetime is future-one connected. Furthermore, it is important, that the boundary of the spacetime is spacelike in the conformal Penrose diagram.\\
Since the whole future of the point is mapped in the chart below the graph and we have uniform causality bounds on the metric, we see, that there is a uniform upper bound on the length of the shortest possible curves on a Cauchy hypersurfcaes that connect any two points in the future of the point (that was mentioned above) intersected with this Cauchy hypersurface.\\
This however contradicts the geometry of the Schwarzschild spacetime. Note, that the hypersurfaces of constant $r$ are Cauchy hypersurfaces. Furthermore, the Schwarzschild spcaetime has a divergent spacelike diameter (compare Definition 5.1 in \cite{Schwarzschild2016}). We will however not use this definition in this thesis (except in this introduction) since it is not quite the right notion for what we want. However, we can find a sequence of such Cauchy hypersurfaces that converge to the curvature singularity $r=0$. It follows by the openness of the future, that for any such Cauchy hypersurface (close enough to $r=0$) we can find two points in the future of the point mentioned above such that their $t$-coordinates differ by a fixed constant, in the chart, for all Cauchy hypersurfaces. The length of the shortest possible curve connecting them on the Cauchy hypersurface then blows up when taking Cauchy hypersurfaces close to $r=0$. Which is a contradiction to the observation above.\\
\\
So to conclude, what are the important properties of the maximal analytically extended Schwarzschild spacetime we need for the proof?

\begin{enumerate}
    \item global hyperbolicity
    \item future one-connectedness
    \item geodesic either hits the curvature singularity or is (future) complete
    \item space-coordinates $t$ and coordintaes of the sphere converge for a timelike curve hitting the curvature singularity
    \item isometries coming from spherical symmetry and $\partial_t$ being static
    \item spacelike boundary at the curvature singularity (in a conformal sense)
    \item divergent spacelike diameter
\end{enumerate}

We will prove the properties 1. and 2. for the Kasner spacetime (Lemmas \ref{globallyhyperbolic} and \ref{futureoneconnected}). To get similar properties as 3.-7., we will need to restrict to a special case of Kasner spacetimes, namely the ones with a \textit{negative Kasner exponent}. It turns out these are the only Kasner spacetimes that have a curvature singularity and all other Kasner spacetimes are smoothly extendible to Minkowski space.\\
We will be able to prove analogues to properties 3. and 4. for these spacetimes with negative Kasner expoenent.
Although not spherical symmetric, the Kasner spacetime has sufficiently many isometries to achieve property 5.\\
Property 6. is easy to see and property 7. can be also proved with the isometries of the Kasner spacetime.

\subsection{Notations and conventions}
We fix some notations and conventions here that will be used throughout this thesis, unless explicitly stated otherwise.\\
\begin{itemize}
    \item We will assume Einstein summation convention $A^iB_i=\sum_{i} A^iB_i$
    \item Greek indices will be assumed to refer to all dimensions, i.e. $\mu,\nu,\kappa,\rho,...\in\{0,1,...,d\}$
    \item Indices that only refer to space dimensions will be written in Latin letters, i.e. $i,j,k,...\in\{1,...,d\}$
    \item All manifolds are considered to be Hausdorff, second countable and of dimension $d+1\geq 2$. Moreover, all manifolds are assumed to be smooth\footnote{We can always assume this, since $M$ must be endowed with an at least $C^1$ differentiable structure, for it to carry a continuous Lorentzian metric. However, then there exists a compatible smooth differentiable structure on $M$ (compare e.g. \cite{hirsch2012differential}).}.
    \item Coordinate vector fields $\partial_\mu:=\frac{\partial
    }{\partial x^\mu}$ on $U\subset M$, for a chart $\varphi:U\rightarrow\mathbb{R}^{d+1}$, are characterized by the fact that $T_p\varphi\cdot\partial_\mu(p) = (\varphi(p),e_\mu)$ for all $p\in U$, where $T_p\varphi$ is the tangent map of $\varphi$ at $p$ and $e_\mu$ is the $\mu$-th standard basis vector of $\mathbb{R}^{d+1}$.
    \item for a curve $\gamma:I\rightarrow \mathbb{R}^{d+1}$ we write $\gamma(s)=(\gamma_t(s),\underline{\gamma}(s))$, where $\gamma_t:I\rightarrow\mathbb{R}$ and $\underline{\gamma}(s)=(\underline{\gamma}_1(s),...,\underline{\gamma}_d(s))\in\mathbb{R}^d$
    \item $\eta=-dx_0^2+\sum_{i=1}^d dx_i^2 = \textup{diag}(-1,1,...,1)$ is the usual Minkowski metric and $d+1$-dimensional Minkowski space is denoted by $(\mathbb{R}^{d+1}_1,\eta)$
\end{itemize}

%% file: Theory.tex
\section{Theory}

\subsection{Definitions and continuous Lorentzian metrics}
In this section, we collect some definitions and prove important results in causality theory with continuous Lorentzian metrics. This will be similar to Chapter 2 in \cite{Schwarzschild2016}. We will assume knowledge of manifolds with smooth Lorentzian metric as introduced for example in \cite{O'Neill}.
\begin{definition}
    Let $(M,g)$ be a Lorentzian manifold with continuous metric. A \textup{timelike curve} is a piecewise smooth curve $\gamma:I\rightarrow M$, where $I\subseteq\mathbb{R}$ is some interval, if for all $s\in I$, where $\gamma$ is differentiable, we have that $\dot{\gamma}(s)$ is timelike and at each point of $I$ the left-sided and right-sided derivative lie in the same connected component of the timelike double cone in the tangent space.\\
    Analogously, we call a piecewise smooth curve $\gamma:I\rightarrow M$ \textup{causal} if $\dot{\gamma}(s)$ is timelike or null in the above sense.
\end{definition}

\begin{definition}
    Let $(M,g)$ be a Lorentzian manifold with continuous metric.
    \begin{itemize}
        \item A \textup{time orientation} is a map $\xi:M\rightarrow P(TM)$, where $P(TM)$ denotes the power set of the tangent bundle $TM$, such that for all $p\in M$: $\xi(p)$ is one of the connected components of the timelike double cone in the tangent space and there exists a chart around $p$ such that $\frac{\partial}{\partial x^0}(q)\in\xi(q)$ for all $q$ in the domain of the chart.\\
        We call a Lorentzian manifold \textup{time-oriented}, a time orientation is chosen. 
        \item A timelike curve $\gamma:I\rightarrow M$ is called \textup{future (past) directed}, if $\dot{\gamma}(s)\in\xi(\gamma(s))$ \,$(-\dot{\gamma}(s)\in\xi(\gamma(s)))$ for all $s\in I$. We will call such curves \textup{FDTL (PDTL)} or \textup{future (past) directed timelike}.
        \item For $p,q\in M$ we define \textup{$p\ll q$ $(p\gg q)$}, if there exists a future (past) directed timelike curve from $p$ to $q$.
        \item The \textup{timelike future} of $p\in M$ in $M$ is defined by $I^+(p,M):=\{q\in M\,|\,p\ll q\}$. The \textup{timelike past} of $p$ in $M$ is defined by $I^-(p,M):=\{q\in M\,|\,p\gg q\}$.
        \item A \textup{spacetime} is a smooth, connected and time-oriented Lorentzian manifold.
    \end{itemize}
\end{definition}

Lorentzian manifolds with merely continuous metric do not posses an exponential map in general. So we are not able to use normal-coordinates, Fermi-coordinates or others that we usually have at our disposal when dealing with smooth (or at least $C^2$) metrics. However, given a timelike curve we can find useful coordinates even when the metric is only continuous. These are also known as cylindrical neighborhoods (Def 1.8 in \cite{continuousmetrics}).

\begin{lemma}\label{chart}
    Let $(M,g)$ be a Lorentzian manifold with a continuous metric $g$, and let $\gamma:[0,1]\rightarrow M$ be a timelike curve. After possibly reparameterising $\gamma$, we get for every $\delta>0$ an open neighborhood $U$ of $\gamma(0)$, some $\varepsilon>0$ and a chart $\varphi:U\rightarrow(-\varepsilon,\varepsilon)^{d+1}$ such that
    \begin{itemize}
        \item $\varphi(\gamma(0))=(0,...,0)$
        \item $(\varphi\circ\gamma)(s)=(s,0...,0)$ for all $s\in[0,\varepsilon)$
        \item $g_{\mu\nu}(0)=\eta_{\mu\nu}$ where $\eta_{\mu\nu}=diag(-1,1,...,1)$
        \item $|g_{\mu\nu}(x)-\eta_{\mu\nu}|<\delta$ for all $x\in(-\varepsilon,\varepsilon)^{d+1}$
    \end{itemize}
\end{lemma}

\begin{proof}
    After a possible linear reparametrization of $\gamma$, we can assume without loss of generality, that
    \begin{equation}\label{gram}
        g(\dot{\gamma}(0),\dot{\gamma}(0))=-1
    \end{equation}
    Since $\gamma$ is piecewise smooth, it has only finitely many breakpoints. So we can find a neighborhood $U$ of $\gamma(0)$ and an $\epsilon>0$ such that $\gamma([0,\epsilon))\subseteq U$ and $\gamma|_{[0,\epsilon)}$ is smooth. Thus $\gamma|_{[0,\epsilon)}$ is a smooth curve with timelike (so nowhere vanishing) tangent vector, i.e. an immersion. By the local immersion theorem (e.g. Theorem 4.12 in \cite{lee2003introduction}) we can find a chart (choosing $U$ and $\epsilon$ smaller if necessary) such that the first two points are satisfied.\\
    In these coordinates (\ref{gram}) is still satisfied. Note that the Lorentzian metric restricted to the space-coordinates is a positive definite inner product. Thus, we get a linear map, obtained from the Gram-Schmidt orthonomalisation procedure based at the origin, to also achieve the third point.\\
    Since the metric is continuous, for a given $\delta>0$ we can choose $U$ and $\epsilon$ even smaller such that the fourth point of the Lemma also holds.    
\end{proof}

These coordinates and similar charts will be used extensively throughout this thesis.\\
\\
We define here some notations that we will often use going forward. Fix $0<a<1$ and let $<.,.>_{\mathbb{R}^{d+1}}$ be the Euclidean inner product, $\|.\|_{\mathbb{R}^{d+1}}$ the induced norm on $\mathbb{R}^{d+1}$ and $e_0=(1,0,...,0)\in\mathbb{R}^{d+1}$. Then we define
\begin{itemize}
    \item $C^+_a:=\{X\in\mathbb{R}^{d+1}\,|\,\frac{<X,e_0>_{\mathbb{R}^{d+1}}}{\|X\|_{\mathbb{R}^{d+1}}}>a\}$
    \item $C^-_a:=\{X\in\mathbb{R}^{d+1}\,|\,\frac{<X,e_0>_{\mathbb{R}^{d+1}}}{\|X\|_{\mathbb{R}^{d+1}}}<-a\}$
    \item $C^c_a:=\{X\in\mathbb{R}^{d+1}\,|\,-a<\frac{<X,e_0>_{\mathbb{R}^{d+1}}}{\|X\|_{\mathbb{R}^{d+1}}}<a\}$
\end{itemize}
     Here, $C^+_a$ describes the forward cone of vectors which form an angle less than $cos^{-1}(a)$ with the $x_0$-axis. Note that the forward and backward cones of timelike vectors in Minkowski space correspond to the value $a=cos(\frac{\pi}{4})=\frac{1}{\sqrt{2}}$.\\
     Furthermore, it is easy to see that if $0<a<\frac{1}{\sqrt{2}}<b<1$, then $C^+_b\subseteq C^+_{\nicefrac{1}{\sqrt{2}}}\subseteq C^+_a$ and $C^-_b\subseteq C^-_{\nicefrac{1}{\sqrt{2}}}\subseteq C^-_a$.

\begin{proposition}\label{futureisopen}
    Let $(M,g)$ be a Lorentzian manifold with a continuous metric $g$. The timelike future $I^+(p,M)$ and timelike past $I^-(p,M)$ are open in $M$ for all $p\in M$.
\end{proposition}

\begin{proof}
    We show that $I^-(p,M)$ is open in $M$, the proof for $I^+(p,M)$ follows analogously.
    Let $q\in I^-(p,M)$ and let $\gamma:[0,1]\rightarrow M$ be a FDTL curve with $\gamma(0)=q$ and $\gamma(1)=p$. We can use Lemma \ref{chart}, to get an according chart $\varphi:U\rightarrow(-\varepsilon,\varepsilon)^{d+1}$ such that $\varphi(q)=(0,...,0)$. Furthermore, since $\frac{5}{6}>\frac{1}{\sqrt{2}}$ it holds that $C^-_{\nicefrac{5}{6}}\subseteq C^-_{\nicefrac{1}{\sqrt{2}}}$. So we can choose $\delta>0$ so small, such that all vectors in $C^-_{\nicefrac{5}{6}}$ are timelike and past directed for all points in the chart $\varphi$.\\
    It is now easy to see that we can find a neighborhood of $(0,...,0)$ in $(\frac{\epsilon}{2},0,...,0)+C^-_{\nicefrac{5}{6}}$.
    This corresponds via the chart to some neighborhood $V\subseteq U$ of $q$ and for all $x\in V$ the straight line from $\varphi(x)$ to $(\frac{\epsilon}{2},0,...,0)$ corresponds to a smooth timelike curve from $x$ to $\gamma(\frac{\epsilon}{2})$ in $U$. Concatenating this curve with $\gamma|_{(\frac{\varepsilon}{2},1]}$ results in a FDTL curve from $x$ to $p$, i.e. $V\subseteq I^-(p,M)$, which concludes the proof.
\end{proof}

This proof can be written down in a more elegant way, but since we will use similar arguments in the future, it is a good place here to introduce them.\\
\\
In the following we will illustrate the usefulness of the notation introduced above. Since $\frac{5}{8}<\frac{1}{\sqrt{2}}<\frac{5}{6}$, we have $C^+_{\nicefrac{5}{6}}\subseteq C^+_{\nicefrac{1}{\sqrt{2}}}\subseteq C^+_{\nicefrac{5}{8}}$. Consider a chart $\Tilde{\varphi}$ as in Lemma \ref{chart} and denote by $R_{\varepsilon_0,\varepsilon_1}:=(-\varepsilon_0,\varepsilon_0)\times(-\varepsilon_1,\varepsilon_1)^d$ the image of this chart. Choose $\delta_0>0$ so small, such that in the chart $\Tilde{\varphi}$ all vectors in $C^\pm_{\nicefrac{5}{6}}$ are timelike and all vectors in in $C^c_{\nicefrac{5}{8}}$ are spacelike.
We prove the following estimates for all $x\in R_{\varepsilon_0,\varepsilon_1}$:\\
    \begin{equation}\label{causalitybounds+}
        (x+C^+_{\nicefrac{5}{6}})\cap R_{\varepsilon_0,\varepsilon_1} \subseteq I^+(x,R_{\varepsilon_0,\varepsilon_1})\subseteq (x+C^+_{\nicefrac{5}{8}})\cap R_{\varepsilon_0,\varepsilon_1}
    \end{equation}
    \begin{equation}\label{causilitybounds-}
        (x+C^-_{\nicefrac{5}{6}})\cap R_{\varepsilon_0,\varepsilon_1} \subseteq I^-(x,R_{\varepsilon_0,\varepsilon_1})\subseteq (x+C^-_{\nicefrac{5}{8}})\cap R_{\varepsilon_0,\varepsilon_1}
    \end{equation}
    \\
    We will only show the second inclusion relation of (\ref{causalitybounds+}), the first one follows trivially since $C^+_{\nicefrac{5}{6}}$ is always included in the light cone of the metric. (\ref{causilitybounds-}) then follows by reversing the time orientation. This was also proven in Theorem 3.1 in \cite{Schwarzschild2016}.\\
    Let $\sigma:[0,L]\rightarrow R_{\varepsilon_0,\varepsilon_1}$ be a FDTL curve with $\sigma(0)=x$ that is parameterised by the arc-length with respect to the Euclidean metric on $\mathbb{R}^{d+1}$, i.e. $L=L_{Euclidean}(\sigma)>0$ is the Euclidean length of the curve $\sigma$. First, we assume that $\sigma$ is smooth. We know that $\dot{\sigma}(s)\in C^+_{\nicefrac{5}{8}}$ for all $s\in [0,L]$, since $\sigma$ is FDTL. Then\\
    \[\frac{<\sigma(L)-x,e_0>_{\mathbb{R}^{d+1}}}{\|\sigma(L)-x\|_{\mathbb{R}^{d+1}}} = \frac{\int^L_0 <\sigma(s),e_0>_{\mathbb{R}^{d+1}}ds}{\|\sigma(L)-x\|_{\mathbb{R}^{d+1}}}> \frac{5}{8}\frac{L}{\|\sigma(L)-x\|_{\mathbb{R}^{d+1}}}\geq \frac{5}{8}\]
    
    where we have used that $\sigma$ is parameterised by arc-length, so $L\geq \|\sigma(L)-x\|_{\mathbb{R}^{d+1}}$. Thus we get $\sigma(L)\in x+C^+_{\nicefrac{5}{8}}$. In the case where $\sigma$ is piecewise smooth we can split the integral into a sum of integrals of smooth segments of $\sigma$. This proves (\ref{causalitybounds+}) and (\ref{causilitybounds-}).

\begin{proposition}\label{proposition2}
    Let $(M,g)$ be a Lorentzian manifold with continuous metric $g$ and $\gamma:[0,1]\rightarrow M$ be FDTL. Then it holds:
    \[I^+(\gamma(0),M)=\bigcup_{0<s\leq 1} I^+(\gamma(s),M)\]
\end{proposition}

\begin{proof}
    The inclusion \textquotedblleft$\supseteq$\textquotedblright\, is obvious.
    For \textquotedblleft$\subseteq$\textquotedblright:\, let $p\in I^+(\gamma(0),M)$, so $\gamma(0)\in I^-(p,M)$. However, the past of $p$ is open in $M$ by Proposition \ref{futureisopen}. Since $\gamma$ is continuous there exists some $s\in (0,1]$ (close to $0$), such that $\gamma(s)\in I^-(p,M)$ and thus $p\in I^+(\gamma(s),M)$.
\end{proof}

\begin{definition}
    Let (M,g) be a time-oriented Lorentzian manifold with continuous metric.
    The \textup{Lorentzian length} of a future directed causal curve $\gamma:[0,1]\rightarrow M$ is \[L(\gamma)=\int^1_0 \sqrt{-g(\dot{\gamma}(s),\dot{\gamma}(s))}ds\]
    The \textup{time separation} or \textup{Lorentzian distance function} $\tau:M\times M\rightarrow[0,\infty]$ is defined by
    \[\tau(p,q)=\left\{\begin{array}{ll} \sup\{L(\gamma)\,|\,\gamma:[0,1]\rightarrow M \, future\,directed, \gamma(0)=p, \gamma(1)=q  &  if\,\,q\in I^+(p,M)\\
    0 & if\,\,q\notin I^+(p,M)\end{array}\right.\]
\end{definition}

\begin{definition}
    Let (M,g) be a time-oriented Lorentzian manifold with continuous metric and let $I\subseteq\mathbb{R}$ be one of the intervals $(a,b),(a,b],[a,b)$ with $a<b$.
    We call a FDTL $\gamma:I\rightarrow M$ \textup{future (past) extendible} if $b\notin I$ $(a\notin I)$ and $\gamma$ can be extended to $I\cup\{b\}$ $(I\cup\{a\})$ as a continuous curve.
    Moreover, $\gamma$ is called \textup{inextendible} if it is future and past inextendible.
\end{definition}

Note that this definition allows for a FDTL curve to be future extendible, but only as a continuous curve and not as a timelike curve.

\begin{definition}
    A subset $S\subseteq M$ of a time-oriented Lorentzian manifold with continuous metric is called \textup{achronal}, if every inextendible timelike curve meets it at most once.
\end{definition}

\begin{definition}
    Let (M,g) be a time-oriented Lorentzian manifold with continuous metric.
    \begin{itemize}
        \item A \textup{Cauchy hypersurface} in $(M,g)$ is a smooth embedded hypersurface $\Sigma$ which is met by every inextendible timelike curve exactly once.\\
        The terms Cauchy hypersurface and Cauchy surface are used interchangeably.
        \item $(M,g)$ is called \textup{globally hyperbolic} if there exists a Cauchy surface $\Sigma$ in $(M,g)$.
    \end{itemize}
\end{definition}

Note that there are multiple equivalent definitions of global hyperbolicity. Since we will only use the existence of a Cauchy hypersurface going forward, this definition suffices.\\
The following two definitions are not particularly well known. However, they will be useful later, when proving the main theorem \ref{maintheorem}.

\begin{definition}
    Let (M,g) be a time-oriented Lorentzian manifold with continuous metric.
    \begin{itemize}
        \item A \textup{timelike homotopy with fixed endpoints} between two FDTL curves $\gamma^0,\gamma^1:[a,b]\rightarrow M$, with $\gamma^0(a)=\gamma^1(a)$ and $\gamma^0(b)=\gamma^1(b)$ $(a<b\in\mathbb{R})$ is a continuous map $\Gamma:[0,1]\times[a,b]\rightarrow M$ such that $\Gamma(0,\cdot)=\gamma^0(\cdot)$ and $\Gamma(1,\cdot)=\gamma^1(\cdot)$ and $\Gamma(u,\cdot)$ is a FDTL curve from $\gamma^0(a)$ to $\gamma^0(b)$ for all $u\in[0,1]$.\\
        We call two FDTL curves \textup{timelike homotopic with fixed endpoints} if there exists a timelike homotopy with fixed endpoints between them.
        \item We call $(M,g)$ \textup{future one-connected} if for all $p,q\in M$, with $p\ll q$, any two FDTL curves from $p$ to $q$ are timelike homotopic with fixed endpoints.
    \end{itemize}
\end{definition}

\begin{definition}\label{timelike-separation}
    Let (M,g) be a time-oriented Lorentzian manifold with continuous metric and $A,B,K\subseteq M$. We say $K$ \textup{timelike-separates} $A$ and $B$ if, and only if, for every timelike curve $\gamma:[0,1]\rightarrow M$ with $\gamma(0)\in A$ and $\gamma(1)\in B$ there is some $\Tilde{s}\in[0,1]$ with $\gamma(\Tilde{s})\in K$.
\end{definition}

For simplicity, assume that all manifolds are connected.

\begin{definition}
    Let (M,g) be a time-oriented Lorentzian manifold with smooth metric.
    \begin{itemize}
        \item Let $k\in\mathbb{N}_0\cup\{\infty\}$. A smooth isometric embedding $\iota:M\hookrightarrow\Tilde{M}$ is called a \textup{$C^k$-extension} of $(M,g)$, if $\Tilde{M}$ is a Lorentzian manifold of the same dimension as $M$, $\partial\iota(M)\neq\emptyset$ and $\Tilde{g}$ is a $C^k$-regular metric.\\
        By slight abuse of terminology, $\Tilde{M}$ is also sometimes called the extension of M.
        \item We call $(M,g)$ \textup{$C^k$-extendible}, if there exists a $C^k$-extension of $(M,g)$. If no such extension exists, we call $(M,g)$ \textup{$C^k$-inextendible}.
    \end{itemize}
\end{definition}

To get more familiar with the concept of $C^0$-extensions, we list in the following a few examples of $C^0$-extensions:
\begin{itemize}
    \item Minkowski space with one point (a closed subset) removed $(\mathbb{R}^{d+1}_1\setminus\underline{0},\eta)$ can be isometrically embedded into Minkowski space  with the obvious map. As Minkowski space has a smooth metric, $(\mathbb{R}^{d+1}_1\setminus\underline{0},\eta)$ is even $C^\infty$-extendible.
    \item the Kerr metric (which we will not write down here) that models the behaviour of a rotating, uncharged, axially symmetric black hole with a quasispherical event horizon is also $C^\infty$-extendible.
    \item $((0,\infty)\times\mathbb{R},g)$ with $g=e^{2\sqrt{t}}(-dt^2+dx^2)$ is obviously $C^0$-extendible to $[0,\infty)\times\mathbb{R}$. However, computing the scalar curvature of $g$ gives $R=-\frac{1}{2e^{2\sqrt{t}}}\frac{1}{t^{\nicefrac{3}{2}}}$, so the scalar curvature blows up for $t\rightarrow 0$, i.e. it is $C^2$-inextendible.
\end{itemize}

The last result shows that curvature blow ups, that often detects $C^2$-inextendibility, do not help us when proving $C^0$-inextendibility. We need more sophisticated tools.\\
\\
In the following section we will develop such tools and results.

\subsection{Results for $C^0$-extensions}

\begin{lemma}\label{curveleavesM}
    Let $(M,g)$ be a time-oriented Lorentzian manifold with smooth metric. If there exists a $C^0$-extension $\iota:M\hookrightarrow\Tilde{M}$, then there exists a timelike curve $\Tilde{\gamma}:[0,1]\rightarrow \Tilde{M}$ such that $\Tilde{\gamma}([0,1))\subseteq\iota(M)$ and $\Tilde{\gamma}(1)\in\Tilde{M}\setminus\iota(M)$.
\end{lemma}

\begin{proof}
    By definition we have $\partial\iota(M)\neq\emptyset$, let $p\in\partial\iota(M)\subseteq\Tilde{M}\setminus\iota(M)$. Choose a small neighborhood $\Tilde{U}$ of $p$ which is time-oriented. For $q\in I^-(p,\Tilde{U})$ we have two cases
    \begin{enumerate}
        \item $q\in\iota(M)$: Then there exists a FDTL curve $\Tilde{\gamma}:[0,1]\rightarrow \Tilde{U}$ with $\Tilde{\gamma}(0)=q\in\iota(M)$ and $\Tilde{\gamma}(1)=p\in\partial\iota(M)$. Since $\iota(M)\subseteq\Tilde{M}$ is open, we see that for
        \[s_0 := \sup_{s\in[0,1]}\,\big\{\,\Tilde{\gamma}([0,s))\subseteq\iota(M)\big\}\]
        we get $\Tilde{\gamma}(s_0)\in\Tilde{M}\setminus \iota(M)$. This gives the timelike curve from the statement after a possible reparametrization.
        \item $q\in\Tilde{M}\setminus\iota(M)$: Since $p\in I^+(q,\Tilde{U})\cap\partial\iota(M)$ and since $I^+(q,\Tilde{U})$ is open, there exists some $r\in I^+(q,\Tilde{U})\cap\iota(M)$. So there exists a PDTL curve $\Tilde{\gamma}:[0,1]\rightarrow \Tilde{U}$ such that $\Tilde{\gamma}(0)=r$ and $\Tilde{\gamma}(1)=q$ and we get the desired timelike curve analogous to the previous case.
    \end{enumerate}
\end{proof}

Note that the proof of the previous lemma shows that if there exists a $C^0$-extension, then there exists a timelike curve leaving the original manifold. However, whether this curve is future or past directed is a priori not clear. 

\begin{definition}
    We call the \textup{future boundary} of a $C^0$-extension $\iota:M\hookrightarrow\Tilde{M}$ the set
    $\partial^+\iota(M)$, which consists of all $p\in\partial\iota(M)$ for which there exists a smooth timelike curve $\gamma:[0,1]\rightarrow\Tilde{M}$ with $\gamma([0,1))\subset\iota(M)$, $\iota^{-1}\circ\gamma:[0,1)\rightarrow M$ is FDTL and $\gamma(1)=p$.
\end{definition}

The past boundary $\partial^-\iota(M)$ is defined analogously.\\
By Lemma \ref{curveleavesM} it is clear that $\partial^+\iota(M)\cup \partial^-\iota(M)\neq\emptyset$.
Furthermore, we clearly have  $\partial^+\iota(M)\cup \partial^-\iota(M)\subseteq\partial\iota(M)$. However, we can not make any general statements about $\partial^+\iota(M)\cap\partial^-\iota(M)$.

\begin{definition}
    We call a $C^0$-extension $\iota:M\hookrightarrow\Tilde{M}$ of a time-oriented Lorentzian manifold with smooth metric \textup{future $C^0$-extension} or \textup{past $C^0$-extension}, if $\partial^+\iota(M)\neq\emptyset$ or $\partial^-\iota(M)\neq\emptyset$ respectively.\\
    As usual, we call a time-oriented Lorentzian manifold with smooth metirc \textup{$C^0$-future/past-inextendible}, if there exists no future/past-extension.
\end{definition}

Lemma \ref{curveleavesM} clearly implies that $M$ is $C^0$-inextendible if, and only if, $M$ is future and past-inextendible.

\subsection{$C^0$-extensions of globally hyperbolic Lorentzian manifolds}
Now, we turn our attention to the important case of $C^0$-extensions of globally hyperbolic time-oriented Lorentzian manifolds. The following proof follows Proposition 1 in \cite{Schwarzschild2018}.

\begin{lemma}\label{lemmagraph}
    Let $\iota:M\hookrightarrow \Tilde{M}$ be a $C^0$-extension of a globally hyperbolic, time-oriented Lorentzian manifold $(M,g)$ with $p\in\partial^-\iota(M)$. For every $\delta>0$ there exists $\varepsilon_0,\varepsilon_1>0$ and a chart $\Tilde{\varphi}:\Tilde{U}\rightarrow (-\varepsilon_0,\varepsilon_0)\times(-\varepsilon_1,\varepsilon_1)^d$ such that:
    \begin{enumerate}
        \item $p\in\Tilde{U}$ and $\Tilde{\varphi}(p)=(0,...,0)$
        \item $|\Tilde{g}_{\mu\nu}-\eta_{\mu\nu}|<\delta$
        \item There exists a Lipschitz continuous function $f:(-\varepsilon_1,\varepsilon_1)^d\rightarrow(-\varepsilon_0,\varepsilon_0)$ such that:\\
        \begin{equation}\label{graph1}
            \{(x_0,\underline{x})\in(-\varepsilon_0,\varepsilon_0)\times(-\varepsilon_1,\varepsilon_1)^d\,|\, x_0>f(\underline{x})\}\subseteq \Tilde{\varphi}(\iota(M)\cap\Tilde{U})
        \end{equation}
        and \begin{equation}\label{graph2}
            \{(x_0,\underline{x})\in(-\varepsilon_0,\varepsilon_0)\times(-\varepsilon_1,\varepsilon_1)^d\,|\, x_0=f(\underline{x})\}\subseteq \Tilde{\varphi}(\partial^-\iota(M)\cap\Tilde{U})
        \end{equation}
        Moreover, the graph of f is achronal in $(-\varepsilon_0,\varepsilon_0)\times(-\varepsilon_1,\varepsilon_1)^d$.
    \end{enumerate}
\end{lemma}

Going forward we will use the abbreviation $R_{\varepsilon_0,\varepsilon_1}:=(-\varepsilon_0,\varepsilon_0)\times(-\varepsilon_1,\varepsilon_1)^d$

\begin{proof}
    By definition of the past boundary and $p\in\partial^-\iota(M)$, there exists a timelike curve $\Tilde{\gamma}:[0,1]\rightarrow\Tilde{M}$ such that $\Tilde{\gamma}((0,1])\subseteq \iota(M)$, $\Tilde{\gamma}|_{(0,1]}$ is FDTL and $\Tilde{\gamma}(0)=p\in\partial^-\iota(M)$.
    We can reparameterise the curve by Lemma \ref{chart} for all $\delta>0$ and find a chart $\Tilde{\varphi}:\Tilde{U}\rightarrow (-\varepsilon_0,\varepsilon_0)\times(-\varepsilon_1,\varepsilon_1)^d$ such that
    \begin{itemize}
        \item $(\Tilde{\varphi}\circ\Tilde{\gamma})(s)=(s,0,..,0)$ for all $s\in[0,\varepsilon_0)$
        \item $|\Tilde{g}_{\mu\nu}-\eta_{\mu\nu}|<\delta$
    \end{itemize}

    Since $\frac{5}{6}>\frac{1}{\sqrt{2}}$, we can choose $\delta>0$ so small that the backwards cone $C^-_{\nicefrac{5}{6}}$ is always contained in the light cone of $\Tilde{g}$, i.e. all vectors in $C^-_{\nicefrac{5}{6}}$ are timelike for all points of the chart $\Tilde{\varphi}$.\\
    \\
    Since $M$ is globally hyperbolic there exists a Cauchy hypersurface $\Sigma\subseteq M$. Note that $\gamma:=\iota^{-1}\circ\Tilde{\gamma}|_{(0,1]}$ is a past inextendible timelike curve in $M$, so we can find a $s_0>0$ close to 0 such that $\gamma(s_0)\in I^-(\Sigma,M)$. Since $I^-(\Sigma,M)$ is open we can choose $\varepsilon_0,\varepsilon_1>0$ smaller, if necessary, such that $[s_0,\varepsilon_0)\times(-\varepsilon_1,\varepsilon_1)^d\subseteq \Tilde{\varphi}(\iota(I^-(\Sigma,M))\cap\Tilde{U})$.\\
    By assuming that $\varepsilon_1<\frac{1}{2\sqrt{d}}\varepsilon_0$ we can guarantee that for all $\underline{x}\in(-\varepsilon_1,\varepsilon_1)^d$ the straight line connection $(-\frac{9}{10}\varepsilon_0,\underline{x})$ with $\underline{0}$ is timelike. Indeed, even the straight line connecting $(-\frac{9}{10}\varepsilon_0,\varepsilon_1,..,\varepsilon_1)$ and $\underline{0}$ is timelike since

    \[\frac{5}{6}<\frac{\frac{9}{10}\varepsilon_0}{\sqrt{(\frac{9}{10}\varepsilon_0)^2+d\varepsilon_1^2}} \iff \varepsilon_1<\sqrt{\frac{11}{25}}\frac{9}{10} \frac{1}{\sqrt{d}} \varepsilon_0\]
    where the right hand side is definitely true, since even $\varepsilon_1<\frac{1}{2\sqrt{d}}\varepsilon_0$ holds, this shows that $(\frac{9}{10}\varepsilon_0,\varepsilon_1,...,\varepsilon_1)\in C^+_{\nicefrac{5}{6}}$.\\
    \\
    Now define $f:(-\varepsilon_1,\varepsilon_1)^d\rightarrow [-\varepsilon_0,\varepsilon_0)$ as follows
    \[f(\underline{x})=\inf_{s_0\in(-\varepsilon_0,\varepsilon_0)} \{(s,\underline{x})\in\Tilde{\varphi}(\iota(I^-(\Sigma,M)) \cap \Tilde{U})\,\,\,\,\forall\, s\in(s_0,\varepsilon_0)\}\]

    First, we show that $f(\underline{x})>-\varepsilon_0$ for all $\underline{x}\in(-\varepsilon_1,\varepsilon_1)^d$.\\
    We argue by contradiction: assume there exists a $\underline{x}_0\in(-\varepsilon_1,\varepsilon_1)^d$ such that $f(\underline{x}_0)=-\varepsilon_0$.
    By the way we defined $f(\underline{x})$ it holds that $(-\frac{9}{10}\varepsilon_0,\underline{x})\in\Tilde{\varphi}(\iota(I^-(\Sigma,M)) \cap \Tilde{U})$ and from the discussion above we know that the straight line $\sigma=(\sigma_0,\underline{\sigma}):[0,1]\rightarrow R_{\varepsilon_0,\varepsilon_1}$ connecting $\sigma(0)=(-\frac{9}{10}\varepsilon_0,\underline{x})$ and $\sigma(1)=0$ is timelike.\\
    We want to show that $\sigma|_{[0,1)}\subseteq \Tilde{\varphi}(\iota(I^-(\Sigma,M)) \cap \Tilde{U})$.\\
    For this we (partially) foliate the plane
    \[\{(t,\underline{\sigma}(s))\in R_{\varepsilon_0,\varepsilon_1}\,|\,t\geq\sigma_0(s)\,,\,s\in[0,1)\}\]
    by (closed) straight lines $\rho_\tau$ that start in $(\tau,\underline{x}_0)$ and end in $\sigma$ with slope $(-\frac{9}{10}\varepsilon_0,-\underline{x}_0)$.\\
    \\
    Note that $(-\frac{9}{10}\varepsilon_0,\underline{x})\in\Tilde{\varphi}(\iota(I^-(\Sigma,M)) \cap \Tilde{U}) \subseteq \Tilde{\varphi}(\iota(M) \cap \Tilde{U})$ and since $\Tilde{\varphi}(\iota(M) \cap \Tilde{U})$ is open, there is some $\Tilde{\tau}\in(-\frac{9}{10}\varepsilon_0,\frac{9}{10}\varepsilon_0)$ such that
    \begin{equation}\label{rho-tau}
        \rho_\tau\subseteq \Tilde{\varphi}(\iota(M) \cap \Tilde{U})\,\,\,\forall \tau\in(-\frac{9}{10}\varepsilon_0,\Tilde{\tau})
    \end{equation}
    And let $\tau_0:=\sup\{\Tilde{\tau}\in(-\frac{9}{10}\varepsilon_0,\frac{9}{10}\varepsilon_0)\,|\,(\ref{rho-tau})\,holds\}$. Note that $\tau_0=\frac{9}{10}\varepsilon_0$.\\
    This holds since $\tau_0<\frac{9}{10}\varepsilon_0$ implies that there exists some $q=(q_0,\underline{q})$ on $\rho_{\tau_0}$ such that $q\notin \Tilde{\varphi}(\iota(M)) \cap \Tilde{U})$. However since lines with slope $(-\frac{9}{10}\varepsilon_0,-\underline{x}_0)$ are timelike it follows that the straight line connecting $(-\frac{9}{10}\varepsilon_0,\underline{x})$ to $q$ corresponds to an future directed timelike and future inextendible  curve entirely contained in $\Tilde{\varphi}(\iota(I^-(\Sigma,M)) \cap \Tilde{U})$, which is a contradiction to $\Sigma$ being a Cauchy surface. See Figure \ref{fig:graph} on the next page.\\
    Thus $\tau_0=\frac{9}{10}\varepsilon_0$, which shows that $\sigma|_{[0,1)}\subseteq \Tilde{\varphi}(\iota(I^-(\Sigma,M)) \cap \Tilde{U})$. This however is also a future directed timelike and future inextendible curve in  $\Tilde{\varphi}(\iota(I^-(\Sigma,M)) \cap \Tilde{U})$, so we get the same contradiction. This proves that $f$ really maps into $(-\varepsilon_0,\varepsilon_0)$.\\

    \begin{figure}[!htp]
    \centering
    \includegraphics[width=0.43\columnwidth]{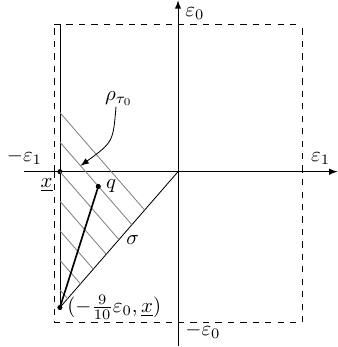}
    \caption{The foliation and curve}
    \label{fig:graph}
    \end{figure}
    
    The properties (\ref{graph1}) and (\ref{graph2}) are obvious by the way we defined $f$.\\
    We use a similar argument as before to show that $f$ is continuous: Let $\underline{x}_n\in(-\varepsilon_0,\varepsilon_0)^d$ be a converging sequence $\underline{x}_n\rightarrow\underline{x}_\infty \in(-\varepsilon_0,\varepsilon_0)^d$. If $f(\underline{x}_n)\nrightarrow f(\underline{x}_\infty)$, then there exists a $\delta>0$ (and possibly a subsequence) such that $|f(\underline{x}_n)-f(\underline{x}_\infty)|>\delta$ for all $n\in\mathbb{N}$.\\
    After possibly restricting again to a subsequence we can without loss of generality assume that $f(\underline{x}_n)<f(\underline{x}_\infty)-\delta$, the case $f(\underline{x}_n)>f(\underline{x}_\infty)+\delta$ follows analogously.
    For $n$ big enough we can connect $(f(\underline{x}_n)+\frac{\delta}{2},\underline{x}_n)$ to $(f(\underline{x}_\infty),\underline{x}_\infty)$ via a straight line that is future directed timelike. This gives the previous contradiction.\\
    The same argument can be also use to show that $f$ is even Lipschitz continuous: Assume $f$ is not Lipschitz, i.e. for all $n\in\mathbb{N}$ there are $x_n,y_n\in(-\varepsilon_0,\varepsilon_0)$ such that $|f(x_n)-f(y_n)|>n|x_n-y_n|$. So the straight line connecting $(f(x_n),x_n)$ and $(f(y_n),y_n)$ has slope more than $n$ (or less than $-n$). Since the straight lines from $\underline{0}$ to $(\frac{9}{10}\varepsilon_0,\underline{x})$ are always timelike, for all $\underline{x}\in(-\varepsilon_1,\varepsilon_1)^d$, we can find some $n\in\mathbb{N}$ such that the line connecting $(f(x_n)+\delta,x_n)$ to $(f(y_n),y_n)$ is timelike for some $\delta>0$. This gives the usual contradiction.\\
    We can even use this argument to prove that the graph of $f$ is achronal in $R_{\varepsilon_0,\varepsilon_1}$. Assume it is not achronal, i.e. there are $x,y$ in the graph of $f$ such that $x=(x_0,\underline{x})\in I^-(y, R_{\varepsilon_0,\varepsilon_1})$. Since $I^-(y, R_{\varepsilon_0,\varepsilon_1})$ is open, there is a $\delta>0$ such that $(x_0+\delta,\underline{x})\in I^-(y, R_{\varepsilon_0,\varepsilon_1})\cap\Tilde{\varphi}(\iota(I^-(\Sigma,M))\cap\Tilde{U})$.
    This allows us to again construct a future directed timelike, future inextendible curve curve in $I^-(\Sigma,M)$, so we get the same contradiction.
\end{proof}

\subsection{Auxiliary results}
The following theorem was proven in \cite{GallowayLingSbierski}.

\begin{theorem}
    Let $(M,g)$ be a globally hyperbolic Lorentzian manifold. Assume there exists a past $C^0$-extension $\iota:M\hookrightarrow\Tilde{M}$, i.e. $\partial^-\iota(M)\neq\emptyset$. Then there exists a FDTL and past inextendible geodesic $\tau:(0,1]\rightarrow M$ such that $\lim_{s\rightarrow 0} (\iota\circ\tau)(s)$ exists and is contained in $\partial\iota(M)$.
\end{theorem}

As noted in \cite{Schwarzschild2018}, we can combine the proof in \cite{GallowayLingSbierski} with Lemma \ref{lemmagraph} to get

\begin{theorem}\label{gallowaylingsbierski}
    Let $\iota:M\hookrightarrow\Tilde{M}$ be a past $C^0$-extension, i.e. $\partial^-\iota(M)\neq\emptyset$, of a globally hyperbolic Lorentzian manifold $(M,g)$. Choose $p\in\partial^-\iota(M)$ and let $\Tilde{\varphi}:\Tilde{U}\rightarrow R_{\varepsilon_0,\varepsilon_1}$ be a chart around $p$ as in Lemma \ref{lemmagraph}. Then there exists a FDTL and past inextendible geodesic $\tau:(0,1]\rightarrow M$ such that $\Tilde{\varphi}\circ\iota\circ\tau:(0,1]\rightarrow R_{\varepsilon_0,\varepsilon_1}$ lies above the graph, i.e. in $\{(s,\underline{x})\in R_{\varepsilon_0,\varepsilon_1}\,|\,s>f(x)\}$, and has endpoint on $\{(s,\underline{x})\in R_{\varepsilon_0,\varepsilon_1}\,|\,s=f(x)\}$.
\end{theorem}

\begin{proof}
    We will only sketch the proof idea here. Let $p\in\partial^-\iota(M)$. By Lemma \ref{curveleavesM} there exists a timelike curve $\Tilde{\gamma}:[0,1]\rightarrow\Tilde{M}$ such that $\Tilde{\gamma}(0)=p$ and $\Tilde{\gamma}|_{(0,1]}\subseteq\iota(M)$ is FDTL. By Lemma \ref{lemmagraph} we get a chart $\Tilde{\varphi}:\Tilde{U}\rightarrow R_{\varepsilon_0,\varepsilon_1}$, for $\varepsilon_1<\frac{1}{2\sqrt{d}}\varepsilon_0$, such that
    \begin{enumerate}
        \item $\Tilde{\varphi}\circ\Tilde{\gamma}(s)=(s,0,...,0)$ for all $s\in[0,\varepsilon_0)$
        \item $|\Tilde{g}_{\mu\nu}-\eta_{\mu\nu}|<\delta$ where $\delta>0$ will be fixed below
        \item There exists a Lipschitz continuous function $f:(-\varepsilon_1,\varepsilon_1)^d\rightarrow(-\varepsilon_0,\varepsilon_0)$, with achronal graph in $(-\varepsilon_0,\varepsilon_0)\times(-\varepsilon_1,\varepsilon_1)^d$, such that:\\
        \begin{equation*}
            \{(x_0,\underline{x})\in(-\varepsilon_0,\varepsilon_0)\times(-\varepsilon_1,\varepsilon_1)^d| x_0>f(\underline{x})\}\subseteq \Tilde{\varphi}(\iota(M)\cap\Tilde{U})
        \end{equation*}            
        and \begin{equation*}
            \{(x_0,\underline{x})\in(-\varepsilon_0,\varepsilon_0)\times(-\varepsilon_1,\varepsilon_1)^d| x_0=f(\underline{x})\}\subseteq \Tilde{\varphi}(\partial^-\iota(M)\cap\Tilde{U})
        \end{equation*}
    \end{enumerate}

    Let $g^{(2)}:=-\frac{1}{4}dx_0^2+\sum_{i=1}^d dx_i^2$ and $g^{(\nicefrac{1}{2})}:=-4dx_0^2+\sum_{i=1}^d dx_i^2$ and pick $\delta>0$ so small, that the following causality bounds hold for any $X\in T\Tilde{U}$:
    \begin{equation*}
    \begin{split}
        g^{(2)}(X,X)\leq 0 & \Longrightarrow g(X,X)<0\\
        g(X,X)\leq 0 & \Longrightarrow g^{(\nicefrac{1}{2})}(X,X)<0
    \end{split}
    \end{equation*}

    We will also note this by $g^{(2)}\prec g\prec g^{(\nicefrac{1}{2})}$.\\
    \\
    Now, choose some $\varepsilon\in(0,\varepsilon_1)$ and put
    \[V_\varepsilon:=I^+_{g^{(\nicefrac{1}{2})}}\big(\Tilde{\varphi}^{-1}(-\frac{\varepsilon}{2},0,...,0),\Tilde{U}\big) \cup I^-_{g^{(\nicefrac{1}{2})}}\big(\Tilde{\varphi}^{-1}(\frac{\varepsilon}{2},0,...,0),\Tilde{U}\big)\]

    As shown in the proof of Theorem 3.3 in \cite{GallowayLingSbierski}, $(V_\varepsilon,\Tilde{g}|_{V_\varepsilon})$ is globally hyperbolic.\\
    Theorem 2.4 in \cite{GallowayLingSbierski} shows that there exists a past directed causal curve $\alpha:[0,2]\rightarrow V_\varepsilon$ from $q:=\Tilde{\varphi}^{-1}(\frac{1}{4}\varepsilon,0,...,0)$ to $p=\Tilde{\gamma}(0)$ with maximal length. Since $\alpha$ clearly lies initially in $\iota(M)$ we can write after reparametrization
    \[\alpha(s)=\left\{\begin{array}{ll} \sigma(s), & s\in[0,1)\\
         \lambda(s), & s\in[1,2]\end{array}\right.\]
    where $\sigma([0,1))\subseteq\iota(M)$ and $\lambda(1)\in\partial\iota(M)$.\\
    $\sigma$ is a maximizer in $V_\varepsilon\cap\iota(M)$, since $\alpha$ is a maximizer in $V_\varepsilon$. But $\Tilde{g}$ is a smooth metric in $\iota(M)$ and since radial geodesics are unique maximizing curves in normal neighborhoods, we see that $\sigma$ is either timelike or null geodesic. If $\sigma$ is a timelike geodesic we are finished. Indeed, it is a PDTL geodesic that is completely contained in $\Tilde{\varphi}^{-1}((-\varepsilon,\varepsilon)^{d+1})\subseteq\Tilde{U}$, has endpoint in $\partial\iota(M)$ and since $0<\varepsilon<\varepsilon_1$ it must end on the graph, i.e. in $\partial^-\iota(M)$.\\
    \\
    Assuming $\sigma$ is a null geodesic yields the same contradiction as in \cite{GallowayLingSbierski}. Note that we need global hyperbolicity of $(M,g)$ for this contradiction. This concludes the proof.
\end{proof}

%% file: Kasner.tex
\section{The Kasner spacetime}

\subsection{Introduction}
The Kasner spacetime $(M,g)$ is a $d+1$-dimensional Lorentzian manifold for $d\geq 3$. It describes an anisotropic expanding universe without matter and is an exact solution to the Einstein vacuum equation ($Ric(g)=0$). It is defined as

    \[M=(0,\infty)\times\mathbb{R}^d\]

with the smooth Lorentzian metric
    \[ g=-dt^2+\sum_{i=1}^{d} t^{2 p_i} \,dx^2_i \]

The so-called \textit{Kasner exponents} have to satisfy the following properties:
\begin{equation}\label{exponentcondition}
    \begin{split}
        & \sum_{i=1}^{d} p_i = 1\\
        & \sum_{i=1}^{d} p_i^2 = 1
    \end{split}
\end{equation}

The first condition defines a $d-1$-dimensional hyperplane and the second conditions defines the standard sphere $\mathbb{S}^{d-1}$. Thus, the possible choices of exponents lie on a $d-2$-dimensional sphere.\\
Note that if the choice of exponents is not one of the \textit{trivial solutions}, i.e. where there exists $j_0\in\{1,...,d\}$ such that $p_{j_0}=1$ and the rest are zero, then there always exists at least one negative Kasner exponent. This follows by combining both condition in (\ref{exponentcondition}):\\
    \[ 1 = (\sum_{i=1}^{d} p_i)^2 = \sum_{i=1}^{d} p_i^2 + 2 \sum_{i=1}^{d-1} p_i \sum_{j>i} p_{j} = 1 + 2 \sum_{i=1}^{d-1} p_i \sum_{j>i} p_{j}\]

and thus
\begin{equation}\label{exponent}
    \sum_{i=1}^{d-1} p_i \sum_{j>i} p_{j} = 0
\end{equation}

Note that (\ref{exponent}) implies that either the solution is one of the trivial ones or there exist at least three non-zero Kasner exponents and at least one (but not all) of them needs to be negative.\\
Moreover, if without loss of generality $p_1<0$, then for $d=3$ we get $-\frac{1}{3}\leq p_1<0$ (by \cite{Misner1973}).\\
\\
Assuming we are in a trivial solution with, without loss of generality, $p_1=1$. Then the metric takes the form
\[g=-dt^2+t^2dx_1^2+\sum_{i=2}^d dx_i^2\]

With the coordinate transformation $\Tilde{t}=t \cosh(x_1)$ and $\Tilde{x}_1=t \sinh(x_1)$ one recovers the Minkowski metric
\[-d\Tilde{t}^2+d\Tilde{x}_1^2 +\sum_{i=2}^d dx_i^2\]

Clearly the range of $\Tilde{t}$ is $(0,\infty)$ and of $\Tilde{x}_1$ is $\mathbb{R}$. We see that the Kasner spacetime with trivial exponents can be isometrically embedded into the open subset $(0,\infty)\times\mathbb{R}^d$ of Minkowski space $\mathbb{R}^{d+1}_1$ with the obvious map. This shows that the trivial case is $C^{\infty}$-extendible.\\
\\
Going forward we will only consider the case where there exists at least one negative Kasner exponent, unless explicitly stated otherwise. This case is sometimes called \textit{Kasner spacetime with negative exponent}.\\

\subsection{Properties}
Since $\partial_t$ is a nowhere vanishing timelike vector field, we can stipulate that it gives the future direction. Thus, the Kasner spacetime is time-oriented, oriented and connected smooth Lorentzian manifold, so it indeed satisfies the requirements of being a spacetime.\\
\\
In the following we will study the curvature of the Kasner spacetime. The only non-vanishing Christoffel symbols are
\begin{equation*}
    \begin{split}
        & \Gamma^0_{ii} =p_it^{2p_i-1}\\
        & \Gamma^i_{0i}=\Gamma^i_{i0}= \frac{p_i}{t}
    \end{split}
\end{equation*}

The Riemann curvature tensor $R\in\mathcal{T}^1_3(M)$ can be calculated in local coordinates as $R_{\partial_\kappa \partial_\mu} \partial_\nu = R_{\kappa \mu \nu}^\sigma \partial_\sigma$ where
\[R_{\kappa\mu\nu}^\sigma:=\partial_\mu \Gamma^\sigma_{\nu\kappa}-\partial_\nu \Gamma^\sigma_{\mu\kappa} + \Gamma^\sigma_{\mu\rho}\Gamma^\rho_{\nu\kappa} - \Gamma^\sigma_{\nu\rho}\Gamma^\rho_{\mu\kappa}\]

The non-vanishing components of the Riemann curvature tensor are with $i\neq j$:
\begin{equation*}
    \begin{split}
        R_{i0i}^0 & =-R_{ii0}^0=p_i(p_i-1)t^{2p_i-2}\\
        R_{00i}^i & =-R_{0i0}^i=p_i(p_i-1)t^{-2}\\
        R_{jij}^i & =-R_{jji}^i=p_ip_jt^{2p_j-2}
    \end{split}
\end{equation*}

We can lower the first index with the metric: $R_{\sigma\kappa\mu\nu}=g^{\sigma\rho}R_{\kappa\mu\nu}^\rho$.
\begin{equation*}
    \begin{split}
        R_{0i0i} & = -R_{0ii0} = R_{i0i0} = -R_{i00i} = p_i(1-p_i)t^{2p_i-2}\\
        R_{ijij} & = -R_{ijji} = p_ip_jt^{2p_i+2p_j-2}
    \end{split}
\end{equation*}

It clearly follows, that the Riemann curvature tensor vanishes if, and only if, the Kasner exponents are trivial solutions.\\
\\
If we are not in the trivial case, i.e. there exists at least one negative Kasner exponent, then the curvature blows up as $t\rightarrow 0$. To see this, we need to study coordinate invariant curvature scalars. In mathematical General Relativity a usual choice for this is the Kretschmann scalar $K=R_{\sigma\kappa\mu\nu}R^{\sigma\kappa\mu\nu}$, where $R^{\sigma\kappa\mu\nu}=g^{\sigma\alpha}g^{\kappa\beta}g^{\mu\gamma}g^{\nu\delta}R_{\alpha\beta\gamma\delta}$.\\
The non-vanishing terms are with $i\neq j$:
\begin{equation*}
    \begin{split}
        R^{0i0i} & = -R^{0ii0} = R^{i0i0} = -R^{i00i} = p_i(1-p_i)t^{-2p_i-2}\\
        R^{ijij} & = -R^{ijji} = p_ip_jt^{-2p_i-2p_j-2}
    \end{split}
\end{equation*}

And thus we see that the Kretschmann scalar of the Kasner spacetime is given by
\[K=\frac{4}{t^4}\Big(\sum_{i=1}^d p_i^2(1-p_i)^2 + \sum_{i=1}^d \sum_{j>i} p_i^2p_j^2\Big)\]

It is easy to conclude that the Kretschmann scalar blows up for $t\rightarrow 0$ if, and only if, we are not in one of the trivial solutions of (\ref{exponentcondition}), i.e. the Kasner spacetime with negative exponent has a curvature singularity for $t\rightarrow 0$. This immediately shows that it is $C^2$-inextendible.\\
\\
The Ricci curvature of a semi-Riemannian manifold is defined to be the contraction $\mathcal{C}^1_3 R\in\mathcal{T}^0_2(M)$ of the Riemann curvature tensor $R$. In local coordinates it can be computed via
\[R_{ij}:=\partial_\mu \Gamma^\mu_{ij}-\partial_j \Gamma^\mu_{i\mu} + \Gamma^\mu_{\mu\nu}\Gamma^\nu_{ij} - \Gamma^\mu_{i\nu}\Gamma^\nu_{\mu j}\]

After a short calculation we get that the only non-zero terms are
\begin{equation*}
    \begin{split}
        & R_{00}=(\sum_{j=1}^d p_j^2 - \sum_{j=1}^d p_j)t^{-2}\\
        & R_{ii}=(1 - \sum_{j=1}^d p_j)p_it^{2p_i-2}
    \end{split}
\end{equation*}

Thus, the two conditions (\ref{exponentcondition}) for the Kasner exponents are satisfied if, and only if, the Ricci curvature vanishes. Then, we see that the Kasner spacetime is indeed a vacuum solutions of the Einstein Equations, i.e. $Ric(g)=0$.\\
\\
\\
The volume element can be computed as 
\[\sqrt{-|g|}=t^{p_1+...+p_d}=t\]

where $|g|$ denotes the determinant of the metric. Since the volume of the spacial slices is always proportional to the volume element, we get that the spatial volume is $O(t)$. So for $t\rightarrow 0$ we can interpret the Kasner spacetime as a cosmological model with a Big Bang (or Big Crunch if we reverse the time orientation).\\
\\
It is noteworthy that an isotropic expansion (or contraction depending on time orientation) is not possible, since the Kasner exponents can't all be equal. Assuming, by contradiction, they are all equal, the first condition would imply that for all $j\in\{1,...,d\}$ we get $p_j=\frac{1}{d}$. However, the second condition then can't be satisfied since $d\geq 3$:
\[\sum_{i=1}^{d} p_i^2 = \frac{1}{d}\neq 1\]
We say the Kasner spacetime models an anisotropic expanding (or contracting) universe.\\
\\
However, the Kasner spacetime still admits a lot of isometries. This can be seen by the fact that for all $i\in\{1,...,d\}$ the coordinate vector fields $\partial_i$ are Killing vector fields. Indeed
\[(\mathcal{L}_{\partial_i} g)_{\mu\nu} = \partial_i g_{\mu\nu} = 0\]

This shows that any one-parameter group of diffeomorphisms $F_{i}:\mathbb{R}\times M\rightarrow M$ with infinitesimal generator $\partial_i$ is a one-parameter group of isometires. These isometries are of the form
\[F_{i,\lambda}(t,x_1,...,x_i,...,x_d)=(t,x_1,...,x_i+\lambda,...,x_d)\]
for some $\lambda\in\mathbb{R}$, where $F_{i, \lambda}(\cdot):=F_i(\lambda,\cdot)$.\\
\\
Since $\partial_i$ are Killing vector fields we know that $g(\partial_i,\dot{\gamma})$ is constant along any geodesics $\gamma$. This implies that a timelike geodesic $\gamma:I\rightarrow M$ in Kasner must satisfy
\begin{equation}\label{geodesicKasner}
    -\dot{\gamma}_t(s)^2 + \sum^d_{i=1} \gamma_t(s)^{-2p_i} k_i^2 = -1
\end{equation}

where $k_i=g(\partial_i,\dot{\gamma})$ are constants.\\
\\
\\
In the following we prove important properties of the Kasner spacetime that we will need for the proof of the $C^0$-inextendibility later.

\begin{lemma}\label{globallyhyperbolic}
    The Kasner spacetime is globally hyperbolic.
\end{lemma}

\begin{proof}
    Note that $\tau:M\rightarrow (0,\infty)$, given by $\tau((t,\underline{x}))=t$, is a smooth temporal function, meaning $grad\,\tau=-\partial_t$ is a PDTL vector field on $M$.\\
    We want to show that for a fixed $t_0\in(0,\infty)$ the set $\Sigma_{t_0}:=\{t=t_0\}$ is a Cauchy surface.\\
    Let $I$ be some Interval and $\gamma:I\rightarrow M$ be FDTL. Note that for all $s\in I$:
    \[0>g(\dot{\gamma}(s),\dot{\gamma}(s))= -\dot{\gamma_t}(s)^2 + \sum_{i=1}^{d} \gamma_t(s)^{2 p_i} \underline{\dot{\gamma}}_i(s)^2 \geq -\dot{\gamma_t}(s)^2\]
    since $\gamma_t(s)>0$ for all $s\in I$. It follows that $|\dot{\gamma_t}(s)|>0$ for all $s\in I$ and by the Inverse Function Theorem that we can always parameterise any timelike curve with respect to the $t$-coordinate, i.e. write the curve as $\gamma(s)=(s,\underline{\gamma}(s))$.\\
    \\
    Since $\tau$ is a temporal function we know that each causal curve that intersects $\Sigma_{t_0}$ only does so exactly once. Thus, it is left to show that any inextendible timelike curve intersects $\Sigma_{t_0}$ at least once.\\
    Fix $0<a<t_0$ and let $\gamma:(a,b)\rightarrow M$ be a FDTL and future inextendible curve with $\gamma(s)=(s,\underline{\gamma}(s))$.
    If $b=\infty$ the curves intersects $\Sigma_{t_0}$ by continuity.\\
    So assume $b<\infty$. Then for all $s\in(a,b)$:
    \[0>g(\dot{\gamma}(s),\dot{\gamma}(s))=-1 + \sum_{i=1}^{d} s^{2 p_i} \underline{\dot{\gamma}}_i(s)^2 \geq -1 +  C_{a,b} \sum_{i=1}^{d} \underline{\dot{\gamma}}_i(s)^2\]
    where \[C_{a,b}= \min_{i=1,..,d}\{\min_{s\in[a,b]} s^{2p_i}\}>0\]
    exists since $0<a<b<\infty$. Thus, we get the following uniform bound for $s\in(a,b)$:
    \[\|\underline{\dot{\gamma}}(s)\|_{\mathbb{R}^d} < \frac{1}{\sqrt{C_{a,b}}}<\infty\]

    Since $\gamma$ was assumed to be in future inextendible, we assumed that for some converging sequence $b_n\rightarrow b$ as $n\rightarrow\infty$ the limit $\lim_{n\rightarrow\infty} \gamma(b_n)$ does not exist. By the way we parameterised $\gamma$ and $b_n$ being convergent, this means that we assumed that the limit of $(\underline{\gamma}(b_n))_{n\in\mathbb{N}}\subseteq\mathbb{R}^d$ does not exist. However, for $n<m\in\mathbb{N}$ we get:
    \[\|\underline{\gamma}(b_m)-\underline{\gamma}(b_n)\|_{\mathbb{R}^d} \leq \int_{b_n}^{b_m} \|\underline{\dot{\gamma}}(s)\|_{\mathbb{R}^d}\,ds < |b_m-b_n|\,\frac{1}{\sqrt{C_{a,b}}}\]
    So $(\underline{\gamma}(b_n))_{n\in\mathbb{N}}$ is a Cauchy sequence in $\mathbb{R}^d$ and thus a limit exists, which is a contradiction to the assumption. Thus, it holds that any such FDTL and future inextendible curve is of the form $\gamma:(a,\infty)\rightarrow M$ and thus intersects $\Sigma_{t_0}$ at $\gamma(t_0)$.\\
    The proof for PDTL and past inextendible curves starting in the future of $\Sigma_{t_0}$ follows analogously. Thus, any inextendible timelike curve is of the form $\gamma:(0,\infty)\rightarrow M$ and meets $\Sigma_{t_0}$ exactly once. We conclude that $\Sigma_{t_0}$ is a Cauchy hypersurface and since $\tau$ is a smooth temporal function it is even smooth and spacelike. This shows that $M$ is globally hyperbolic.
\end{proof}

\begin{lemma}\label{geodesicto0/infinity}
    Let $\tau:[a,b)\rightarrow M$ be a future or past directed timelike geodesic that is future or past inextendible respectively. Then it holds that $(t\circ\tau)(s)\rightarrow \infty$ or $(t\circ\tau)(s)\rightarrow 0$ as $s\rightarrow b$ respectively.
\end{lemma}

\begin{proof}
    This follows from the proof of Lemma \ref{globallyhyperbolic}.
\end{proof}

\begin{lemma}\label{futureoneconnected}
    The Kasner spacetime is future one-connected.
\end{lemma}

\begin{proof}
    Let $\gamma^0,\gamma^1:[a,b]\rightarrow M$ be two FDTL curves with same start and end point, i.e. $\gamma^0(a)=\gamma^1(a)$ and $\gamma^0(b)=\gamma^1(b)$. Note that we can again assume that $\gamma^i(s)=(s,\underline{\gamma}^i(s))$ for $i\in\{0,1\}$. One can then define a homotopy $\Gamma:[0,1]\times[a,b]\rightarrow M$ by
    \[\Gamma(u,s):=(s,(1-u)\gamma^0(s)+u\gamma^1(s))\]
    Indeed one has $\Gamma(i,s)=\gamma^i(s)$, $\Gamma(u,a)=\gamma^i(a)$ and $\Gamma(u,b)=\gamma^i(b)$ for $i\in\{0,1\}$ and all $u\in[0,1]$.
    So it is left to show that $\Gamma(u,\cdot)$ is FDTL for all $u\in[0,1]$.\\
    Using the fact that $x\mapsto x^2$ is convex, $\gamma^0$ and $\gamma^1$ are timelike and $s^{2 p_i}>0$ for all $0<a\leq s\leq b<\infty$ and $i\in\{1,...,d\}$ we get for all $u\in[0,1]$:
    \begin{equation*}
        \begin{split}
            g(\partial_s\Gamma(u,s),\partial_s\Gamma(u,s)) & = -1 + \sum_{i=1}^d s^{2 p_i} ((1-u)\underline{\dot{\gamma}}^0_i(s)+u\underline{\dot{\gamma}}^1_i(s))^2\\
            & \leq -1 + (1-u)\sum_{i=1}^d s^{2 p_i} \underline{\dot{\gamma}}^0_i(s)^2 + u\sum_{i=1}^d s^{2 p_i} \underline{\dot{\gamma}}^1_i(s)^2\\
            & = (1-u)g(\dot{\gamma}^0(s),\dot{\gamma}^0(s)) + ug(\dot{\gamma}^1(s),\dot{\gamma}^1(s))<0
        \end{split}
    \end{equation*}
    
    This proves that the Kasner spacetime is future one-connected.
\end{proof}

\begin{lemma}\label{kasnerfuturecomplete}
    The Kasner spacetime is future complete, i.e. any FDTL affinely parameterised and future inextendible geodesic is of the form $\gamma:(a,\infty)\rightarrow M$.
\end{lemma}

\begin{proof}
    Let $\gamma:I\rightarrow M$ be a FDTL and future inextendible geodesic. Assume without loss of generality that it is parameterised by arc length (unit speed).\\
    By (\ref{geodesicKasner}) we know that it must satisfy
    \[\dot{\gamma}_t(s)^2 = \sum^d_{i=1} \gamma_t(s)^{-2p_i} k_i^2 + 1\]

    Furthermore, by Lemma \ref{geodesicto0/infinity} it holds that $\gamma_t(s)\rightarrow\infty$. The claim is proven once we show that this implies $s\rightarrow\infty$.\\
    So we want to show, that $\gamma_t(s)$ doesn't blow up in finite time.
    Assume without loss of generality that $\gamma_t(s)\geq 1$ and let $C:=\max \{\,k^2_i\,|\, i=1,...,d\,\}$.
    Since we know that $-1<p_i$ for all $i\in\{1,...,d\}$, $\gamma_t(s)^{-2p_i}\leq \gamma_t(s)^2$ holds. All together we get\\
    \[\dot{\gamma}_t(s)^2 \leq  Cd \gamma_t(s)^2 + 1\]

    This implies, that for some constant $\Tilde{C}>0$ and $C_d:=\sqrt{Cd}$ we get
    \[\gamma_t(s)\leq \frac{\sinh(C_d\,\Tilde{C}+C_d\,s)}{C_d}\]

    Note that the right hand side goes to $\infty$ if, and only if, $s\rightarrow\infty$. Since any geodesic is a linear reparametrization of a unit speed geodesic geodesic, this proves the claim.
\end{proof}

\begin{lemma}\label{kasnerpastincomplete}
    The Kasner spacetime is past incomplete.
\end{lemma}

\begin{proof}
    This can be calculated directly by constructing a PDTL geodesic that reaches $\{t=0\}$ in finite affine parameter time. Let $p\in M$ and $\gamma:[0,b)\rightarrow M$ be some unit speed PDTL geodesic with $\gamma(0)=p$. Assume the conserved quantities are $g(\partial_i,\dot{\gamma})=0$ for all $i\in\{1,...,d\}$. By (\ref{geodesicKasner}) it then holds that
    \[\dot{\gamma}(s)_t^2 = 1\]
    Thus we have
    \[\gamma(s)=(\gamma_t(0)-s,\underline{\gamma}(0))\]
    So we can reach $\{t=0\}$ in finite affine parameter time and thus $b<\infty$. Which means the Kasner spacetime is not past complete, which completes the proof.
\end{proof}

The following is an important Lemma to study the nature of the Kasner singularity.

\begin{lemma}\label{spacecoordinates}
    For all $\varepsilon>0$ there exists some $T>0$ such that for any timelike curve $\gamma:I\rightarrow M$ (for which there is some $s_0\in I$ such that $\gamma_t(s_0)<T$) and all $i\in\{1,...,d\}$, we have for all $s,\Tilde{s}\in \gamma_t^{-1}((0,T\,])$ 
    \[|\underline{\gamma}_i(s)-\underline{\gamma}_i(\Tilde{s})|<\varepsilon\]
\end{lemma}

\begin{proof}
    First note that this statement is purely geometric, meaning it is invariant under reparametrization of the curve. Since $\gamma$ is timelike we can parameterise it with respect to the $t$-coordinate, i.e. $\gamma:\Tilde{I}\rightarrow M$ with $\gamma(s)=(s,\underline{\gamma}(s))$.\\
    Fix some $i\in\{1,..,d\}$. Then for any $s\in\Tilde{I}$ (so $s>0$) it holds that
    \[0>g(\dot{\gamma}(s),\dot{\gamma}(s))=-1 + s^{2 p_i} \,\underline{\dot{\gamma}}_i(s)^2 + \sum_{j\neq i} s^{2 p_j} \,\underline{\dot{\gamma}}_j(s)^2 \geq -1 + s^{2 p_i} \,\underline{\dot{\gamma}}_i(s)^2\]
    And thus we get
    \[|\underline{\dot{\gamma}}_i(s)|< s^{-p_i}\]
    Since we are in the case where there exists a negative Kasner exponent, we know that $p_i<1$ holds for all $i\in\{1,..,d\}$. Thus, the upper bound is integrable on $(0,T\,]$, for all $T>0$, and the Lemma follows from integration.
\end{proof}

Note that the previous lemma implies that for all FDTL $\gamma:(0,b)\rightarrow M$ with $\gamma(s)=(s,\underline{\gamma}(s))$ and $i\in\{1,...,d\}$, we have that $\underline{\gamma}_i(s)$ converge as $s\searrow 0$.

\section{The main theorem}
The proof follows the strategy established for Theorem 1 in \cite{Schwarzschild2018}.\\
By Lemma \ref{curveleavesM} we know that if there exists a $C^0$-extension of the Kasner spacetime, then $\partial^+\iota(M)\cup\partial^-\iota(M)\neq\emptyset$. So we deal with the two cases separately.

\begin{theorem}\label{Kasnerfutureinextendible}
    The Kasner spacetime (with a negative exponent) is future $C^0$-inextendible.
\end{theorem}

\begin{proof}
    We prove this by contradiction. Assume there exists a future $C^0$-extension $\iota:M\hookrightarrow\Tilde{M}$ of the Kasner spacetime $(M,g)$, i.e. $\partial^+\iota(M)\neq\emptyset$. By combining Lemma \ref{lemmagraph} (with reversed time orientation or directly Proposition 1 from \cite{Schwarzschild2018}) and Theorem \ref{gallowaylingsbierski}, there exists $0<\varepsilon_1<\frac{1}{2\sqrt{d}}\varepsilon_0$, a chart $\Tilde{\varphi}:\Tilde{U}\rightarrow (-\varepsilon_0,\varepsilon_0)\times(-\varepsilon_1,\varepsilon_1)^d$ and a FDTL geodesic $\tau:[0,1)\rightarrow M$ that is future inextendible in M such that:
    \begin{itemize}
        \item $|\Tilde{g}_{\mu\nu}-\eta_{\mu\nu}|<\delta_0$ (where we will fix $\delta_0>0$ below)
        \item There exists a Lipschitz continuous function $f:(-\varepsilon_1,\varepsilon_1)^d\rightarrow(-\varepsilon_0,\varepsilon_0)$ such that:\\
        \begin{equation}\label{graph3}
            \{(x_0,\underline{x})\in R_{\varepsilon_0,\varepsilon_1}\,|\, x_0<f(\underline{x})\}\subseteq \Tilde{\varphi}(\iota(M)\cap\Tilde{U})
        \end{equation}
        and \begin{equation}\label{graph4}
            \{(x_0,\underline{x})\in R_{\varepsilon_0,\varepsilon_1}\,|\, x_0=f(\underline{x})\}\subseteq \Tilde{\varphi}(\partial^-\iota(M)\cap\Tilde{U})
        \end{equation}
        Moreover, the graph of f is achronal in $\Tilde{U}$.
        \item $\Tilde{\varphi}\circ\iota\circ\tau:[0,1)\rightarrow R_{\varepsilon_0,\varepsilon_1}$ maps into $\{(x_0,\underline{x})\in R_{\varepsilon_0,\varepsilon_1}\,|\, x_0<f(\underline{x})\}$ and, after recentering the chart, we can also arrange for $\lim_{s\rightarrow 1} (\Tilde{\varphi}\circ\iota\circ\tau)(s)=(0,...,0)$.
    \end{itemize}

    Since $\frac{5}{8}<\frac{1}{\sqrt{2}}<\frac{5}{6}$, we have $C^+_{\nicefrac{5}{6}}\subseteq C^+_{\nicefrac{1}{\sqrt{2}}}\subseteq C^+_{\nicefrac{5}{8}}$. So we can choose $\delta_0>0$ so small, such that at all points in the image of the chart $\Tilde{\varphi}$ all vectors in $C^\pm_{\nicefrac{5}{6}}$ are timelike and all vectors in $C^c_{\nicefrac{5}{8}}$ are spacelike.\\
    We have already seen before that the following holds for all $x\in R_{\varepsilon_0,\varepsilon_1}$:
    \begin{equation*}
        (x+C^+_{\nicefrac{5}{6}})\cap R_{\varepsilon_0,\varepsilon_1} \subseteq I^+(x,R_{\varepsilon_0,\varepsilon_1})\subseteq (x+C^+_{\nicefrac{5}{8}})\cap R_{\varepsilon_0,\varepsilon_1}
    \end{equation*}
    \begin{equation*}
        (x+C^-_{\nicefrac{5}{6}})\cap R_{\varepsilon_0,\varepsilon_1} \subseteq I^-(x,R_{\varepsilon_0,\varepsilon_1})\subseteq (x+C^-_{\nicefrac{5}{8}})\cap R_{\varepsilon_0,\varepsilon_1}
    \end{equation*}

    These estimates were proven as (\ref{causalitybounds+}) and (\ref{causilitybounds-}) in Chapter 2.\\
    \\
    Since the geodesic $\tau$ is FDTL and future inextendible in $M$, Lemma \ref{geodesicto0/infinity} implies that $(t\circ\tau)(s)\rightarrow \infty$ as $s\rightarrow 1$.\\
    \\
    We claim that any such timelike geodesic has infinite length.\\
    Since $\tau$ is future inextendible, Lemma \ref{kasnerfuturecomplete} implies that $\tau$ is future complete. So after reparameterising to unit speed we can calculate the length
    \[L(\tau)=\int_a^{\infty} d\Tilde{s}=\infty\]

    Since the length is invariant under reparametrization this finished the claim.\\
    \\
    For readability sake we ease notation and denote $\Tilde{\varphi}\circ\iota\circ\tau$ in the following by $\tau$. Because it is timelike we know for all $s\in[0,1)$ we have $\dot{\tau}(s)\in C^+_{\nicefrac{5}{8}}$ and thus $dx_0(\dot{\tau}(s))>0$. So after a reparametrization we can assume $\tau:[s_0,1)\rightarrow R_{\varepsilon_0,\varepsilon_1}$ to be given by $\tau(s)=(s-1,\underline{\tau}(s))$, where $s_0\in (1-\varepsilon_0,1)$. So $\dot{\tau}(s)=e_0+\sum^d_{i=1} \underline{\dot{\tau}}_i(s)e_i$ and it follows that

    \[\frac{5}{8} < \frac{<\dot{\tau}(s),e_0>_{\mathbb{R}^{d+1}}}{\|\dot{\tau}(s)\|_{\mathbb{R}^{d+1}}} = \frac{1}{\sqrt{1+\|\underline{\dot{\tau}}(s)\|^2_{\mathbb{R}^{d+1}}}}\]

    Which shows that for all $s\in[s_0,1)$ we have
    \[\|\underline{\dot{\tau}}(s)\|_{\mathbb{R}^{d+1}}< \frac{\sqrt{39}}{5}\]

    Together with the uniform bound on the metric components, it follows that there exists some constant $C>0$ such that
    \[\int^1_{s_0} \sqrt{-\Tilde{g}(\dot{\tau}(s),\dot{\tau}(s))}\,ds = \int^1_{s_0} \sqrt{-\Big(\Tilde{g}_{00} + 2 \sum^d_{i=1} \Tilde{g}_{0i}\underline{\dot{\tau}}_i(s) + \sum^d_{i,j=1} \Tilde{g}_{ij}\underline{\dot{\tau}}_i(s)\underline{\dot{\tau}}_j(s)\Big)}\,ds \leq C<\infty\]

    so $\tau$ has finite length, which is a contradiction.
\end{proof}

We are left to find a contradiction for the more interesting case, namely that no timelike curve can leave through the curvature singularity at $\{t=0\}$.

\begin{theorem}\label{kasnerpastinextendible}
    The Kasner spacetime (with a negative exponent) is past $C^0$-inextendible.
\end{theorem}

\begin{proof}
    The proof is also by contradiction. Assume there exists a past $C^0$-extension $\iota:M\hookrightarrow\Tilde{M}$ of the Kasner spacetime $(M,g)$. By Lemma \ref{curveleavesM} there exists $\Tilde{\gamma}:[0,1]\rightarrow \Tilde{M}$ such that $\Tilde{\gamma}((0,1])\subseteq\iota(M)$ is FDTL and $\Tilde{\gamma}(0)\in\partial^-\iota(M)$.
    By  combining Lemma \ref{lemmagraph} and Lemma \ref{chart}, there exists $0<\varepsilon_1<\frac{1}{2\sqrt{d}}\varepsilon_0$, a chart $\Tilde{\varphi}:\Tilde{U}\rightarrow (-\varepsilon_0,\varepsilon_0)\times(-\varepsilon_1,\varepsilon_1)^d$ such that:
    \begin{itemize}
        \item $(\Tilde{\varphi}\circ\Tilde{\gamma})(s)=(s,0,...,0)$ for all $s\in[0,\varepsilon_0)$
        \item $|\Tilde{g}_{\mu\nu}-\eta_{\mu\nu}|<\delta_0$ (where $\delta_0>0$ is chosen as in Theorem \ref{Kasnerfutureinextendible})
        \item There exists a Lipschitz continuous function $f:(-\varepsilon_1,\varepsilon_1)^d\rightarrow(-\varepsilon_0,\varepsilon_0)$ such that:\\
        \begin{equation}\label{graph3}
            \{(x_0,\underline{x})\in R_{\varepsilon_0,\varepsilon_1}\,|\, x_0>f(\underline{x})\}\subseteq \Tilde{\varphi}(\iota(M)\cap\Tilde{U})
        \end{equation}
        and \begin{equation}\label{graph4}
            \{(x_0,\underline{x})\in R_{\varepsilon_0,\varepsilon_1}\,|\, x_0=f(\underline{x})\}\subseteq \Tilde{\varphi}(\partial^-\iota(M)\cap\Tilde{U})
        \end{equation}
        Moreover, the graph of f is achronal in $\Tilde{U}$.
    \end{itemize}

    We assume again the same causality bounds as (\ref{causalitybounds+}) and (\ref{causilitybounds-}) in Chapter 2.\\
    Set $\gamma:=\iota^{-1}\circ\Tilde{\gamma}|_{(0,\varepsilon_0)}$, which is a FDTL and past inextendible curve in $M$.\\
    \\
    The rest of the proof proceeds in three steps.
    \begin{addmargin}[5pt]{0pt}
    \underline{\textbf{Step 1:}} Show that there exists a $\mu>0$ such that
    \begin{enumerate}
        \item $\iota(I^-(\gamma(\mu),M)) \subseteq \Tilde{\varphi}^{-1}(\{(x_0,\underline{x})\in R_{\varepsilon_0,\varepsilon_1}\,|\,x_0>f(\underline{x})\})$
        \item $(\frac{49}{50}\varepsilon_0,\varepsilon_0)\times(-\varepsilon_1,\varepsilon_1)^d \subseteq I^+((\Tilde{\varphi}\circ\Tilde{\gamma})(\mu),R_{\varepsilon_0,\varepsilon_1})$
    \end{enumerate}
    \end{addmargin}
    Remember that we assumed that $0<\varepsilon_1<\frac{1}{2\sqrt{d}}\varepsilon_0$. Now choose $x^+:=(x^+_0,0,...,0)$ with $0<x^+_0<\varepsilon_0$ and $x^-:=(x^-_0,0,...,0)$ with $-\varepsilon_0>x^-_0>0$ such that the closure of $(x^++C^-_{\nicefrac{5}{6}})\cap (x^-+C^+_{\nicefrac{5}{6}})$ in $R_{\varepsilon_0,\varepsilon_1}$ is compact.
    Then choose $y^+:=(y^+_0,0,...,0)$ with $0<y^+_0<\frac{1}{5}\varepsilon_0$ such that the closure of $C^+_{\nicefrac{5}{8}}\cap(y^++C^-_{\nicefrac{5}{8}})$ in $R_{\varepsilon_0,\varepsilon_1}$ is contained in $(x^++C^-_{\nicefrac{6}{7}})\cap (x^-+C^+_{\nicefrac{6}{7}})$. In the following we show that these choices imply
    \begin{equation}\label{step1.2done}
        (\frac{49}{50}\varepsilon_0,\varepsilon_0)\times(-\varepsilon_1,\varepsilon_1)^d \subseteq I^+(y^+,R_{\varepsilon_0,\varepsilon_1})
    \end{equation}

    \begin{figure}[!htp]
    \centering
    \includegraphics[width=0.48\columnwidth]{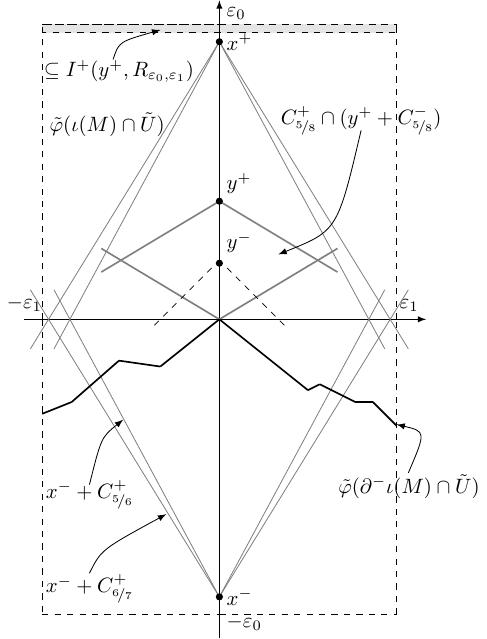}
    \caption{Step 1}
    \label{fig:step1}
    \end{figure}
   
    We check for what $y^+_0>0$ we can guarantee that that the straight line connecting the origin to $(\frac{49}{50}\varepsilon_0-y^+_0,\underline{x})$ is still timelike for all $\underline{x}\in(-\varepsilon_1,\varepsilon_1)^d$. Note that for these $y^+_0$ (\ref{step1.2done}) follows immediately.
    We again want to find bounds that even $(\frac{49}{50}\varepsilon_0-y^+_0,\varepsilon_1,...,\varepsilon_1)$ is timelike, which we can ensure by arranging for $(\frac{49}{50}\varepsilon_0-y^+_0,\varepsilon_1,...,\varepsilon_1)\in C^+_{\nicefrac{5}{6}}$, which would imply the previous statement.
    \[\frac{5}{6}<\frac{\frac{49}{50}\varepsilon_0-y^+_0}{\sqrt{(\frac{49}{50}\varepsilon_0-y^+_0)^2+d\varepsilon_1^2}} \iff \varepsilon_1<\sqrt{\frac{11}{25}}(\frac{49}{50}\varepsilon_0-y^+_0)\frac{1}{\sqrt{d}}\]
    Since we assumed that $\varepsilon_1<\frac{1}{2\sqrt{d}}\varepsilon_0$ the right hand side is definitely a true statement if
    \[\frac{\varepsilon_0}{2} \leq \sqrt{\frac{11}{25}}(\frac{49}{50}\varepsilon_0-y^+_0)\]
    which is equivalent to
    \[y_0^+ \leq (\frac{49}{50}-\frac{1}{2}\sqrt{\frac{25}{11}})\varepsilon_0\]

    So in particular everything works out if we assume the bound $0<y_0^+<\frac{1}{5}\varepsilon_0$. We fix such a $y^+=(y^+_0,0,...,0)$ and note that for all $0<\Tilde{x}_0<y^+_0$ the property (\ref{step1.2done}) is also satisfied by $\Tilde{x}:=(\Tilde{x}_0,0,...,0)$.\\
    \\
    \\
    The proof of Step 1 is based on five claims that we will prove in the following.\\
    \underline{Claim 1:} for all $0<s<y^+_0$ we have
    \begin{equation}\label{regioninchart}
        \Tilde{\varphi}^{-1}\Big(I^+\big((s,0,...,0),R_{\varepsilon_0,\varepsilon_1}\big)\cap I^-(y^+,R_{\varepsilon_0,\varepsilon_1})\Big) = \iota(I^+\big(\gamma(s),M\big)\cap I^-(\gamma(y_0^+),M))
    \end{equation}

    The inclusion "$\subseteq$"\ \,follows from the fact that the graph of $f$ is achronal, so all FDTL curves in $R_{\varepsilon_0,\varepsilon_1}$ that start in the future of the graph (for example in $(s,0,...,0)$ for $s>0$) stay above it. So in particular, they are contained in $\iota(M)$.\\
    For the inclusion "$\supseteq$", let $\sigma:[s,y^+_0]\rightarrow M$ be a FDTL curve from $\gamma(s)$ to $\gamma(y^+_0)$. We want to show that $\iota\circ\sigma$ maps into $\Tilde{U}$.
    There is a timelike homotopy $\Gamma:[0,1]\times [s,y^+_0]\rightarrow M$ with fixed endpoints between $\gamma|_{[s,y_0^+]}$ and $\sigma$ by Lemma \ref{futureoneconnected}. Note that $\iota\circ\Gamma:[0,1]\times [s,y^+_0]\rightarrow \Tilde{M}$ is then a timelike homotopy with fixed endpoints in $\Tilde{M}$. So we want to show that $(\iota\circ\sigma)(\cdot)=(\iota\circ\Gamma)(1,\cdot)$ maps into $\Tilde{U}$.\\
    Let $J:=\{t\in[0,1]\,|\, (\iota\circ\Gamma)(t,[s,y^+_0])\subseteq\Tilde{U}\}$. Since $(\iota\circ\Gamma)(0,\cdot)=\gamma|_{[s,y^+_0]}$, we have $0\in J$, so $J$ is non-empty. Furthermore, since $\Tilde{U}$ is open it follows that $J$ is open in $[0,1]$. We have arranged above that $I^+(\Tilde{\gamma}(s),\Tilde{U})\cap I^-(\Tilde{\gamma}(y^+_0),\Tilde{U})$ is precompact in $\Tilde{U}$. Thus its closure in $\Tilde{M}$ is contained in $\Tilde{U}$, which shows that $J$ is also closed. By the connectedness of $[0,1]$ we get $J=[0,1]$, i.e. $(\iota\circ\sigma)(\cdot)=(\iota\circ\Gamma)(1,\cdot)$ maps into $\Tilde{U}$.\\
    \\
    Together with Proposition \ref{proposition2} we then get that
    \begin{equation}
        \Tilde{\varphi}^{-1}(I^+(0,R_{\varepsilon_0,\varepsilon_1})\cap I^-(y^+,R_{\varepsilon_0,\varepsilon_1}) = \iota\Big(\bigcup_{0<s<y_0^+} I^+(\gamma(s),M)\cap I^-(\gamma(y_0^+),M)\Big)
    \end{equation}

    This allows us to make statements about the causal diamond, i.e. the right hand side, by just computing in the chart.\\
    We fix some $y_0^-\in(0,y_0^+)$, we define the set
    \[K:=\Big[\bigcup_{0<s<y_0^+} I^+(\gamma(s),M)\cap I^-(\gamma(y_0^+),M)\Big]\setminus I^-(\gamma(y_0^-),M)\]

    \begin{figure}[!htp]
    \centering
    \includegraphics[width=0.6\columnwidth]{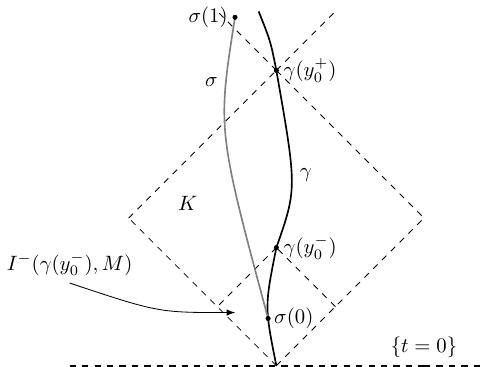}
    \caption{For the proof of claim 2}
    \label{fig:setK}
    \end{figure}

    \underline{Claim 2:} $K$ timelike-separates $\gamma((0,y^-_0))$ from $I^+(\gamma(y_0^+),M)$. (see Def \ref{timelike-separation})\\
    \\
    To prove this, let $\sigma:[0,1]\rightarrow M$ be some FDTL curve with $\sigma(0)\in\gamma((0,y^-_0))$ and $\sigma(1)\in I^+(\gamma(y_0^+),M)$. Since $\sigma$ then starts in $I^-(\gamma(y_0^-),M)$, which is open, we see that $\sigma^{-1}(I^-(\gamma(y_0^-),M))$ is non-empty and open in $[0,1]$. Furthermore, it is connected, because if there is some $s_0\in\sigma^{-1}(I^-(\gamma(y_0^-),M))$, then also $[0,s_0]\subseteq\sigma^{-1}(I^-(\gamma(y_0^-),M))$ follows since $\sigma$ is a FDTL curve.
    Finally, the Kasner spacetime is causal, meaning that there exists no closed timelike curves, since it is globally hyperbolic. It follows that $I^+(\gamma(y_0^+),M)\cap I^-(\gamma(y_0^-),M) \subseteq I^+(\gamma(y_0^-),M)\cap I^-(\gamma(y_0^-),M)=\emptyset$, thus we get $\sigma(1)\notin I^-(\gamma(y_0^-),M)$.\\
    This shows the existence of a $\Delta_-\in(0,1)$ such that $\sigma^{-1}(I^-(\gamma(y_0^-),M))=[0,\Delta_-)$.\\
    Using the same argument we can proof that there exists a $\Delta_+\in(0,1)$ such that
    $\sigma^{-1}(I^-(\gamma(y_0^+),M))=[0,\Delta_+)$.\\
    We claim that $\Delta_-<\Delta_+$. Note that the Kasner spacetime $M$ is globally hyperbolic and has a smooth metric, it holds that $\overline{I^-(\gamma(y_0^-),M)}=J^-(\gamma(y_0^-),M)$ (e.g. in \cite{O'Neill}).\\
    Together with the so-called push-up principle for smooth Lorentzian metrics we get
    \[\overline{I^-(\gamma(y_0^-),M)}=J^-(\gamma(y_0^-),M) \subseteq J^-\big(I^-(\gamma(y_0^+),M)\big) \overset{\textup{push-up}}{=} I^-(\gamma(y_0^+),M)\]
    So $\sigma(\Delta_-)\in\overline{I^-(\gamma(y_0^-),M)}\subseteq I^-(\gamma(y_0^+),M)$ which, combined with the causality of Kasner and $\sigma$ being FDTL, this implies $\Delta_-<\Delta_+$.\\
    For all $s_0\in(\Delta_-,\Delta_+)$, we have $\sigma(s_0)\in I^-(\gamma(y_0^+),M)\setminus I^-(\gamma(y_0^-),M)$. And since $\sigma(0)\in\gamma((0,y_0^-))$, we get $\sigma(s_0)\in \bigcup_{0<s<y_0^+} I^+(\gamma(s),M)$. So in total $\sigma(s_0)\in K$, which proves the claim.\\
    \\
    \\
    \underline{Claim 3:} $\overline{K}$ is compact in $M$.\\
    \\
    Similar to Lemma \ref{spacecoordinates}, we set
    \[S:=\max_{i=1,...,d} {\int_0^{\gamma_t(y_0^+)} s^{-p_i} ds}\]
    Note that since there is at least one negative Kasner exponent, we know that $p_i<1$ for all $i\in\{1,...,d\}$, thus $S<\infty$. This shows that
    \[\overline{K}\subseteq (0,\gamma_t(y_0^+)]\times[\underline{\gamma}_1(y_0^+)-S,\underline{\gamma}_1(y_0^+)+S]\times...\times [\underline{\gamma}_d(y_0^+)-S,\underline{\gamma}_d(y_0^+)+S]\]
    We want to show that there exists a $T>0$ such that for all $t\leq T$ and $(t,\underline{x})\in \big(\bigcup_{0<s<y_0^+} I^+(\gamma(s),M)\cap I^-(\gamma(y_0^+),M)\big)$, we have that $(t,\underline{x})\notin K$, i.e. $(t,\underline{x})\in I^-(\gamma(y_0^-),M)$.\\
    Since $I^-(\gamma(y_0^-),M)$ is open and $\gamma$ is FDTL, it holds that for all $s_0\in(0,y_0^-)$ there exists a $\delta>0$ such that
    \[\{\gamma_t(s_0)\}\times (\underline{\gamma}_1(s_0)-\delta,\underline{\gamma}_1(s_0)+\delta)\times ...\times (\underline{\gamma}_d(s_0)-\delta,\underline{\gamma}_d(s_0)+\delta) \subseteq I^-(\gamma(y_0^-),M)\]

    By the isometries of the Kasner spacetime we know that for all $x_1,...,x_d\in(-\delta,\delta)$ the curve $\sigma:(0,s_0]\rightarrow M$ defined by $\sigma(s)=(\gamma_t(s),\underline{\gamma}_1(s)+x_1,...,\underline{\gamma}_d(s)+x_d)$ is then also timelike. So in total we see that for all $s_0\in(0,y_0^-)$ there exists a $\delta>0$ such that
    \begin{equation}\label{claim-compact}
        \{\gamma_t(s)\}\times (\underline{\gamma}_1(s)-\delta,\underline{\gamma}_1(s)+\delta)\times ...\times (\underline{\gamma}_d(s)-\delta,\underline{\gamma}_d(s)+\delta) \subseteq I^-(\gamma(y_0^-),M)
    \end{equation}
    holds for all $s\in(0,s_0]$.\\
    Choose some $s_0\in(0,y_0^+)$ and let $\delta>0$ such that (\ref{claim-compact}) holds. Without loss of generality, we choose some $T\in(0,\gamma_t(y_0^+))$ such that Lemma \ref{spacecoordinates} holds for $\frac{\delta}{2}$ (i.e. for $\varepsilon=\frac{\delta}{2}$).\\
    Note that if $(T,\underline{x})\in \big(\bigcup_{0<s<y_0^+} I^+(\gamma(s),M)\cap I^-(\gamma(y_0^+),M)\big)$ then there must exist some FDTL curve $\sigma:[0,1]\rightarrow M$ with $\sigma(0)\in\gamma((0,y_0^+))$ and $\sigma(1)=(T,\underline{x})$. Let $s_1,s_2\in(0,y_0^+)$ such that $\gamma(s_1)=\sigma(0)$ and $\gamma_t(s_2)=\sigma_t(1)=T$. It follows that
    \[|\underline{\sigma}_i(1)-\underline{\gamma}_i(s_2)|\leq |\underline{\sigma}_i(1)-\underline{\sigma}_i(0)| + |\underline{\sigma}_i(0)-\underline{\gamma}_i(s_2)| = |\underline{\sigma}_i(1)-\underline{\sigma}_i(0)| + |\underline{\gamma}_i(s_1)-\underline{\gamma}_i(s_2)| < \delta\]
    since Lemma \ref{spacecoordinates} holds for $\sigma$ and $\gamma$. Because $\sigma$ was arbitrary, we get for all $s\in(0,\varepsilon_0)$
    \begin{equation*}
        \begin{split}
            I^+(\gamma(s),M)\cap \big\{ t\leq T\big\} & \subseteq \bigcup_{0<\Tilde{s}\leq s_0} \Big[ \{\gamma_t(\Tilde{s})\}\times (\underline{\gamma}_1(\Tilde{s})-\delta,\underline{\gamma}_1(\Tilde{s})+\delta)\times ...\times (\underline{\gamma}_d(\Tilde{s})-\delta,\underline{\gamma}_d(\Tilde{s})+\delta) \Big]\\
            & \subseteq I^-(\gamma(y_0^-),M)
        \end{split}
    \end{equation*}
    by (\ref{claim-compact}). This implies that
    \[K\subseteq (T,\gamma_t(y_0^+))\times [\underline{\gamma}_1(y_0^+)-S,\underline{\gamma}_1(y_0^+)+S]\times...\times [\underline{\gamma}_d(y_0^+)-S,\underline{\gamma}_d(y_0^+)+S]\]
    and thus
    \[\overline{K}\subseteq [T,\gamma_t(y_0^+)]\times [\underline{\gamma}_1(y_0^+)-S,\underline{\gamma}_1(y_0^+)+S]\times...\times [\underline{\gamma}_d(y_0^+)-S,\underline{\gamma}_d(y_0^+)+S]\]
    is compact, which finishes the claim.\\
    \\
    \\
    Note that since $\overline{K}$ is compact and $\Tilde{\varphi}$ and $\iota$ are continuous, we get
    \[\Tilde{\varphi}(\iota(\overline{K})) = \overline{\Tilde{\varphi}(\iota(K))}\]
    
    By the choice of $y^+\in R_{\varepsilon_0,\varepsilon_1}$ in the beginning of Step 1, (\ref{regioninchart}) and the definition of $K$, we obtain
    \[\overline{\Tilde{\varphi}(\iota(K))} \subseteq \overline{I^+(0,R_{\varepsilon_0,\varepsilon_1})\cap I^-(y^+,R_{\varepsilon_0,\varepsilon_1})} \subseteq (x^++C^-_{\nicefrac{6}{7}})\cap (x^-+C^+_{\nicefrac{6}{7}})\]

    Thus, we can define
    \[W:= (\Tilde{\varphi}\circ\iota)^{-1} \Big((x^++C^-_{\nicefrac{6}{7}})\cap (x^-+C^+_{\nicefrac{6}{7}})\Big) \subseteq M\]
    which is an open neighborhood of $\overline{K}$ in $M$ by the discussion above.\\
    \\
    \\
    \underline{Claim 4:} $\exists\,\mu>0$ such that $I^-(\gamma(\mu),M)$ is timelike separated from $\gamma(x_0^+)$ by $W$.\\
    \\
    We can define the Euclidean metric on the Kasner spacetime $d:M\times M\rightarrow[0,\infty)$. Note that the map $M\rightarrow[0,\infty)$ given by
    \[(t,\underline{x})\mapsto d((t,\underline{x}),M\setminus W)= \inf_{(\Tilde{t},\Tilde{\underline{x}})\in M\setminus W} d((t,\underline{x}),(\Tilde{t},\Tilde{\underline{x}}))\]
    is continuous. Since $\overline{K}$ is compact and disjoint from the closed set $M\setminus W$, we see that this map must attain a minimum $\delta>0$ over $\overline{K}$. This means that
    \begin{equation}\label{K_delta}
        \overline{K}_\delta:= \big\{(t,\underline{x})\in M\,|\,d((t,\underline{x}),\overline{K})<\delta \big\} \subseteq W
    \end{equation}

    and by possibly choosing $\delta>0$ even smaller we can also assume that
    \begin{equation}\label{delta2}
        \{\gamma_t(x_0^+)\}\times(\underline{\gamma}_1(x_0^+)-\delta,\underline{\gamma}_1(x_0^+)+\delta)\times ...\times(\underline{\gamma}_d(x_0^+)-\delta,\underline{\gamma}_d(x_0^+)+\delta) \subseteq I^+(\gamma(y_0^+),M)
    \end{equation}

    We fix this $\delta>0$ such that (\ref{K_delta}) and (\ref{delta2}) hold. By Lemma \ref{spacecoordinates} there exists $\mu\in(0,y_0^-)$ such that for all $(t_0,\underline{x}_0)\in I^-(\gamma(\mu),M)$, we have
    \begin{equation}\label{mu}
        |\underline{x}_0-\underline{\gamma}(\mu)|<\frac{\delta}{2}
    \end{equation}
    
    We show that this implies that $\overline{K}_\delta$ timelike separates $I^-(\gamma(\mu),M)$ from $\gamma(x_0^+)$.\\
    Let $\sigma:[0,1]\rightarrow M$ be FDTL with $\sigma(0)\in I^-(\gamma(\mu),M)$ and $\sigma(1)=\gamma(x_0^+)$. Choose $s_0\in(0,\mu)$ such that $\gamma_t(s_0)=\sigma_t(0)$. From (\ref{mu}) we see that for all$\,i\in\{1,...,d\}$
    \begin{equation}\label{mu2}
        |\underline{\sigma}_i(0)-\underline{\gamma}_i(s_0)| \leq |\underline{\sigma}_i(0)-\underline{\gamma}_i(\mu)| + |\underline{\gamma}_i(\mu)-\underline{\gamma}_i(s_0)| < \frac{\delta}{2}+\frac{\delta}{2} =\delta
    \end{equation}
    Now define a new curve $\hat{\sigma}:[0,1]\rightarrow M$ as
    \[\hat{\sigma}(s):=(\sigma_t(s),\underline{\sigma}_1(s)+\underline{\gamma}_1(s_0)-\underline{\sigma}_1(0),...,\underline{\sigma}_d(s)+\underline{\gamma}_d(s_0)-\underline{\sigma}_d(0))\]
    It easy to see by the isometries of the Kasner spacetime that $\hat{\sigma}$ is also a FDTL curve. Furthermore, this curve starts in $\hat{\sigma}(0)=\gamma(s_0)\in\gamma((0,y_0^-))$ and ends in
    \[\hat{\sigma}(1)=(\sigma_t(1),\underline{\sigma}_1(1)+\underline{\gamma}_1(s_0)-\underline{\sigma}_1(0),...,\underline{\sigma}_d(1)+\underline{\gamma}_d(s_0)-\underline{\sigma}_d(0))\]
    Note that by (\ref{mu2}), we thus have $\hat{\sigma}(1)\in I^+(\gamma(y_0^+),M)$.\\
    By our previous observation, we know that then there exist a $\hat{s}\in[0,1]$ such that $\hat{\sigma}(\hat{s})\in K$. However, this implies that $\sigma(\hat{s})\in \overline{K}_\delta$, so $I^-(\gamma(\mu),M)$ is timelike-separated from $\gamma(x_0^+)$ by $\overline{K}_\delta$.
    Note that this, in particular, implies that $W$ timelike-separates $I^-(\gamma(\mu),M)$ from $\gamma(x_0^+)$ which proves the claim.\\
    \\
    \\
    \underline{Claim 5:} $\iota(I^-(\gamma(\mu),M))\subseteq \Tilde{\varphi}^{-1}\big((x^++C^-_{\nicefrac{6}{7}}) \cap (x^-+C^+_{\nicefrac{6}{7}})\big)$.\\
    \\
    We argue by contradiction. Let $\sigma:[0,1]\rightarrow M$ be PDTL with $\sigma(0)=\gamma(\mu)$ and assume there is some $\Tilde{s}\in[0,1]$ such that $(\Tilde{\varphi}\circ\iota\circ\sigma)(\Tilde{s})\notin (x^++C^-_{\nicefrac{6}{7}}) \cap (x^-+C^+_{\nicefrac{6}{7}})$. Let
    \[s_0:=\sup \{t\in[0,1]\,|\,(\Tilde{\varphi}\circ\iota\circ\sigma)(s)\in (x^++C^-_{\nicefrac{6}{7}}) \cap (x^-+C^+_{\nicefrac{6}{7}})\,\forall s\in[0,t)\}\]

    It is easy to see that $0<s_0\leq 1$ and $(\Tilde{\varphi}\circ\iota\circ\sigma)(s_0)\in \partial\big((x^++C^-_{\nicefrac{6}{7}}) \cap (x^-+C^+_{\nicefrac{6}{7}})\big)$. Since all vectors in $C^+_{\nicefrac{5}{6}}$ are future directed timelike and $C^+_{\nicefrac{6}{7}}\subseteq C^+_{\nicefrac{5}{6}}$ we can find a FDTL curve from $(\Tilde{\varphi}\circ\iota\circ\sigma)(s_0)$ to $x^+$, that does not intersect $(x^++C^-_{\nicefrac{6}{7}}) \cap (x^-+C^+_{\nicefrac{6}{7}})$.
    Take for example a curve that lies on $\partial\big((x^++C^-_{\nicefrac{6}{7}}) \cap (x^-+C^+_{\nicefrac{6}{7}})\big)$ that is depicted in Figure \ref{fig:schluss} on the next page.\\
    
    \begin{figure}[!htp]
    \centering
    \includegraphics[width=0.44\columnwidth]{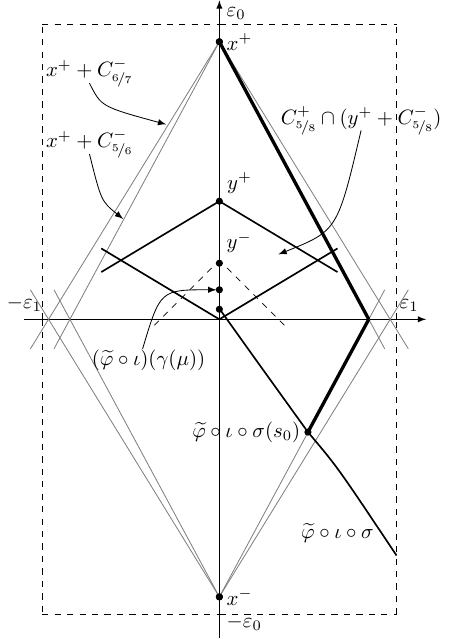}
    \caption{One possible curve}
    \label{fig:schluss}
    \end{figure}
    
    However, this curve corresponds to a FDTL curve from $\sigma(s_0)\in I^-(\gamma(\mu),M)$ to $\gamma(x_0^+)$ in $M$, which does not intersect $W:= (\Tilde{\varphi}\circ\iota)^{-1} \Big((x^++C^-_{\nicefrac{6}{7}})\cap (x^-+C^+_{\nicefrac{6}{7}})\Big)$, which is a contradiction to claim 4 and thus proves claim 5.\\
    \\
    This implies the first point of Step 1 and the second point follows from (\ref{step1.2done}) together with $0<\mu<y_0^-<y_0^+$.\\
    \\
    \begin{addmargin}[5pt]{0pt}
    \underline{\textbf{Step 2:}} We show that Step 1 implies, the existence of a constant $0<C_d<\infty$ such that for any Cauchy hypersurface $\Sigma$ of $M$ and any two points in $I^-(\gamma(\mu),M)\cap\Sigma$ the distance in $\Sigma$ is bounded by $C_d$, i.e. for all $p,q\in I^-(\gamma(\mu),M)\cap\Sigma$ we have
    \[d_\Sigma(p,q):=\inf_{\gamma:[0,1]\rightarrow\Sigma\,\,\text{piecewiese smooth},\,\gamma(0)=p,\,\gamma(1)=q} \Big\{\int_0^1 \sqrt{\overline{g}(\dot{\gamma}(s),\dot{\gamma}(s))}ds\Big\} \leq C_d\]
    where $\overline{g}$ is the induced metric on $\Sigma$.\\
    \end{addmargin}
    
    We only consider Cauchy hypersurface $\Sigma$ in $M$ for which $\gamma(\mu)\in I^+(\Sigma,M)$ holds. There is nothing to prove otherwise. For any $\underline{x}\in(-\varepsilon_1,\varepsilon_1)^d$ we can now consider the curves $\sigma_{\underline{x}}:(f(\underline{x}),\varepsilon_0)\rightarrow M$ given by
    \[\sigma_{\underline{x}}(s)=(\iota^{-1}\circ\Tilde{\varphi}^{-1})(s,\underline{x})\]
    
    Clearly, these are FDTL and past inextendible curves in $M$ and by the second point of Step 1, they end up in $I^+(\gamma(\mu),M)\subseteq I^+(\Sigma,M)$. Each of these curves need to intersect the Cauchy hypersurface $\Sigma$ exactly once. This allows us to define a well defined map
    \[h:(-\varepsilon_1,\varepsilon_1)^d\rightarrow(-\varepsilon_0,\varepsilon_0)\]

    where $h(\underline{x})$ is the unique intersection point of $\sigma_{\underline{x}}$ and $\Sigma$.\\
    \\
    \underline{Claim 1:} $h:(-\varepsilon_1,\varepsilon_1)^d\rightarrow(-\varepsilon_0,\varepsilon_0)$ is smooth.\\
    \\
    Fix some $\underline{x}_0\in(-\varepsilon_1,\varepsilon_1)^d$ and remember that we defined Cauchy hypersurfaces to be smoothly embedded hypersurfaces, thus $\Tilde{\varphi}(\iota(\Sigma)\cap \Tilde{U})$ is a smooth submanifold of $R_{\varepsilon_0,\varepsilon_1}$.
    By definition there exists an open neighborhood $W\subseteq R_{\varepsilon_0,\varepsilon_1}$ of $(h(\underline{x}_0),\underline{x}_0)$ and a smooth submersion $u:W\rightarrow\mathbb{R}$ such that $\Tilde{\varphi}(\iota(\Sigma)\cap \Tilde{U})\cap W = \{u=0\}$.
    Furthermore, the timelike vector field $\partial_0$ can be nowhere tangent to $\Tilde{\varphi}(\iota(\Sigma)\cap \Tilde{U})$, since the tangent spaces of Cauchy surfaces do not contain any timelike vectors, thus $\partial_0 u|_{(h(\underline{x}_0),\underline{x}_0)}\neq 0$.
    It follows from the implicit function theorem that there exists an open neighborhood $V\subseteq (-\varepsilon_1,\varepsilon_1)^d$ of $\underline{x}_0$ and a smooth function $v:V\rightarrow\mathbb{R}$ such that $u(v(\underline{x}),\underline{x})=0$. From the definition of $h$ it must thus hold that $h|_V=v$, which shows that $h$ is smooth.\\
    \\
    \underline{Claim 2:} there exists a $0<C<\infty$ such that $|\partial_i h(\underline{x})|\leq C$ holds for all $\underline{x}\in(-\varepsilon_1,\varepsilon_1)^d$ and all $i\in\{1,...,d\}$.\\
    \\
    Clearly, $T_{(h(\underline{x}),\underline{x})} \Tilde{\varphi}(\iota(\Sigma)\cap \Tilde{U}) = span\{(\partial_1 h|_{\underline{x}})\partial_0+\partial_1,..., (\partial_d h|_{\underline{x}})\partial_0+\partial_d\}$ are the tangent spaces for all $\underline{x}\in(-\varepsilon_1,\varepsilon_1)^d$. Since no timelike vector can be contained in the tangent space of a Cauchy hypersurface, we get for all $i\in\{1,...,d\}$
    \begin{equation}\label{boundforh}
        0\leq \Tilde{g}((\partial_i h)\partial_0+\partial_i, (\partial_i h)\partial_0+\partial_i)= (\partial_ih)^2\Tilde{g}_{00}+2(\partial_ih)\Tilde{g}_{0i}+\Tilde{g}_{ii}
    \end{equation}

    where equality holds if, and only if
    \[(\partial_ih)_\pm= \frac{-\Tilde{g}_{0i}\pm \sqrt{(\Tilde{g}_{0i})^2-\Tilde{g}_{ii}\Tilde{g}_{00}}}{\Tilde{g}_{00}}\]

    which is clearly independent of $h$ and thus of the Cauchy hypersurface $\Sigma$.
    \[\max\{|(\partial_ih)_+|,|(\partial_ih)_-|\}\leq C\]
    holds due to the uniform causality bounds $|\Tilde{g}_{00}|<-1+\delta_0<0$ and $|\Tilde{g}_{\mu\nu}|<1+\delta_0$, where $0<C<\infty$ is a constant depending only on $\delta_0$. Moreover, since $\Tilde{g}_{00}<0$, the inequality (\ref{boundforh}) implies
    \[(\partial_ih)_-<(\partial_ih)<(\partial_ih)_+\]
    and thus $|\partial_ih|\leq C$ for all $i\in\{1,...,d\}$.\\
    \\
    Now we define the graph of $h$ by $\psi:(-\varepsilon_1,\varepsilon_1)^d\rightarrow R_{\varepsilon_0,\varepsilon_1}$ with $\psi(\underline{x})=(h(\underline{x}),\underline{x})$.
    This parameterises a smooth submanifold $\Tilde{S}$ of $R_{\varepsilon_0,\varepsilon_1}$  which is isometric to an open subset $S$ of $\Sigma\subseteq M$, via $\iota^{-1}\circ\Tilde{\varphi}^{-1}$.
    With respect to the chart $\psi^{-1}$, we denote the components of the metric $\overline{g}$ on $\Tilde{S}$ that is induced by $\Tilde{g}$ on $R_{\varepsilon_0,\varepsilon_1}$ by $\overline{g}_{ij}$, where $i,j\in\{1,...,d\}$.\\
    \\
    \underline{Claim 3:} there exists a constant $0<C_{\overline{g}}<\infty$ such that $|\overline{g}_{ij}(\underline{x})|<C_{\overline{g}}$ holds for all $\underline{x}\in(-\varepsilon_1,\varepsilon_1)^d$ and all $i,j=1,...,d$.\\
    \\
    The components of the induced metric can be computed by
    \[\overline{g}_{ij}=(\psi^*\Tilde{g})_{ij}= \Tilde{g}_{\mu\nu}\frac{\partial\psi^{\mu}}{\partial x_i} \frac{\partial\psi^{\nu}}{\partial x_j} = \Tilde{g}_{00}\frac{\partial h}{\partial x_i} \frac{\partial h}{\partial x_j} + \Tilde{g}_{0j}\frac{\partial h}{\partial x_i} + \Tilde{g}_{i0} \frac{\partial h}{\partial x_j} + \Tilde{g}_{ij}\]

    Thus, the uniform causality bounds of the metric combined with the uniform bounds form Claim 2 imply, that there exists a constant $0<C_{\overline{g}}<\infty$ such that $|\overline{g}_{ij}|\leq C_{\overline{g}}$ holds for all $i,j\in\{1,..,d\}$, which finishes the claim.\\
    \\
    To complete Step 2, let $p,q\in I^-(\gamma(\mu),M)\cap\Sigma$. $\iota(I^-(\gamma(\mu),M))$ is completely contained in the image of the chart $\Tilde{\varphi}$, by the first point of Step 1. Thus, there exist $\underline{x},\underline{y}\in(-\varepsilon_1,\varepsilon_1)^d$ such that $\psi(\underline{x})=\Tilde{\varphi}(\iota(p))$ and $\psi(\underline{y})=\Tilde{\varphi}(\iota(q))$.
    Let $\sigma:[0,1]\rightarrow (-\varepsilon_1,\varepsilon_1)^d$ be given by $\sigma(s)=\underline{x}+s(\underline{y}-\underline{x})$, i.e. the straight line in $(-\varepsilon_1,\varepsilon_1)^d$ connecting $\underline{x}$ and $\underline{y}$.\\
    It follows that
    \begin{equation*}
        \begin{split}
            L(\sigma) & =\int_0^1 \sqrt{\overline{g}\big(\dot{\sigma}(s),\dot{\sigma}(s)\big)}\,ds\\
            & = \int_0^1 \sqrt{\Big(\sum_{i,j=1}^d (\underline{y}_i-\underline{x}_i)\overline{g}_{ij}\big(\sigma(s)\big) (\underline{y}_j-\underline{x}_j)\Big)}\,ds\\
            & \leq  \int_0^1 \sqrt{\Big(\sum_{i,j=1}^d \sqrt{d}\varepsilon_1 C_{\overline{g}} \sqrt{d}\varepsilon_1\Big)}\, ds\\
            & =\varepsilon_1 d^{\nicefrac{3}{2}} \sqrt{C_{\overline{g}}}
        \end{split}
    \end{equation*}

    This is a uniform bound that is independent of $\underline{x},\underline{y}\in(-\varepsilon_1,\varepsilon_1)^d$.
    We can connect $p$ and $q$ in $\Sigma$ by the smooth curve $\iota^{-1}\circ\Tilde{\varphi}^{-1}\circ \psi \circ \sigma$, which has length less or equal to $\varepsilon_1 d^{\nicefrac{3}{2}} \sqrt{C_{\overline{g}}}$.\\
    Since $S\subseteq\Sigma$ we get
    \begin{equation*}
        \begin{split}
            d_{\Sigma}(p,q) & = \inf_{\gamma:[0,1]\rightarrow\Sigma\,\,\text{piecewiese smooth},\,\gamma(0)=p,\,\gamma(1)=q} \Big\{\int_0^1 \sqrt{\overline{g}(\dot{\gamma}(s),\dot{\gamma}(s))}ds\Big\}\\
            & \leq \inf_{\gamma:[0,1]\rightarrow S\,\,\text{piecewiese smooth},\,\gamma(0)=p,\,\gamma(1)=q} \Big\{\int_0^1 \sqrt{\overline{g}(\dot{\gamma}(s),\dot{\gamma}(s))}ds\Big\}\\
            & \leq \varepsilon_1 d^{\nicefrac{3}{2}} \sqrt{C_{\overline{g}}}
        \end{split}
    \end{equation*}

    which concludes Step 2 with $C_d:= \varepsilon_1 d^{\nicefrac{3}{2}} \sqrt{C_{\overline{g}}}$.\\
    \\
    \begin{addmargin}[5pt]{0pt}
    \underline{\textbf{Step 3:}} We show that the geometry of $(M,g)$ contradicts Step 2.\\
    \end{addmargin}
    
    As seen in the proof of Lemma \ref{globallyhyperbolic} the hypersurfaces of constant $t$ are Cauchy hypersurfaces, so consider the family $\Sigma_n:=\{t = \frac{1}{n}\}$, $n\in\mathbb{N}$, of Cauchy hypersurfaces. The induced metric $\overline{g}_n$ on $\Sigma_n$ is given by
    \[\overline{g}_n=\sum_{i=1}^d n^{-2p_i} dx_i^2\]

    Note that $\gamma(\frac{\mu}{2})\in I^-(\gamma(\mu),M)$, so by openness of the past we know that there exists some $\delta>0$ such that 
    \[\{\gamma_t(\frac{\mu}{2})\}\times [\underline{\gamma}_1(\frac{\mu}{2})-\delta,\underline{\gamma}_1(\frac{\mu}{2})+\delta]\times ...\times [\underline{\gamma}_d(\frac{\mu}{2})-\delta,\underline{\gamma}_d(\frac{\mu}{2})+\delta] \subseteq I^-(\gamma(\mu),M)\]
    And since $\partial_t$ was a timelike vector field we clearly get
    \[(0,\gamma_t(\frac{\mu}{2})]\times [\underline{\gamma}_1(\frac{\mu}{2})-\delta,\underline{\gamma}_1(\frac{\mu}{2})+\delta]\times ...\times [\underline{\gamma}_d(\frac{\mu}{2})-\delta,\underline{\gamma}_d(\frac{\mu}{2})+\delta] \subseteq I^-(\gamma(\mu),M)\]

    Choose $i_0\in\{1,...,d\}$ that corresponds to a negative Kasner exponent, i.e. $p_{i_0}<0$. Consider the following sequence of points
    \[p_n:=\big(\frac{1}{n},\underline{\gamma}_1(\frac{\mu}{2}),...,\underline{\gamma}_{i_0}(\frac{\mu}{2})-\delta,...,\underline{\gamma}_d(\frac{\mu}{2})\big)\in I^-(\gamma(\mu),M)\cap\Sigma_n\]
    and
    \[p_n:=\big(\frac{1}{n},\underline{\gamma}_1(\frac{\mu}{2}),...,\underline{\gamma}_{i_0}(\frac{\mu}{2})+\delta,...,\underline{\gamma}_d(\frac{\mu}{2})\big)\in I^-(\gamma(\mu),M)\cap\Sigma_n\]

    with $n\geq n_0$ for some sufficiently large $n_0\in\mathbb{N}$.\\
    It is easy to see that the shortest piecewise smooth curve connecting $p_n$ and $q_n$ in $\Sigma_n$ is given by $\gamma_n:[-\delta,\delta]\rightarrow M$
    \[\gamma_n(s)=(\frac{1}{n},\underline{\gamma}_1(\frac{\mu}{2}),...,\underline{\gamma}_{i_0}(\frac{\mu}{2})+s,...,\underline{\gamma}_d(\frac{\mu}{2}))\]

    The length $L(\gamma_n)$ of $\gamma_n$ is given by
    \[L(\gamma_n)=\int_{-\delta}^\delta \sqrt{\overline{g}_n(\dot{\gamma}_n(s),\dot{\gamma}_n(s))}ds=\int_{-\delta}^\delta \sqrt{n^{-2\,p_{i_0}}} = 2\delta \sqrt{n^{-2\,p_{i_0}}}\]
    Since $p_{i_0}<0$ it follows that
    \[d_{\Sigma_n}(p_n,q_n)=2\delta \sqrt{n^{-2\,p_{i_0}}} \rightarrow \infty\]
    as $n\rightarrow\infty$, which contradicts Step 2. This concludes the proof.
\end{proof}

So by combining both previous theorems we finally get the main theorem.\\

\begin{theorem}\label{maintheorem}
    The Kasner spacetime (with a negative exponent) is $C^0$-inextendible.
\end{theorem}

\begin{proof}
    The proof is by contradiction, so assume there exists a $C^0$-extension $\iota:M\hookrightarrow\Tilde{M}$ of the Kasner spacetime. Then by Lemma \ref{curveleavesM} we get $\partial^+\iota(M)\cup\partial^-\iota(M)\neq\emptyset$. However, Theorem \ref{Kasnerfutureinextendible} shows that $\partial^+\iota(M)=\emptyset$ and Theorem \ref{kasnerpastinextendible} shows that $\partial^-\iota(M)=\emptyset$, witch is a contradiction. Thus the Kasner spacetime (with a negative exponent) is $C^0$-inextendible.
\end{proof}

%% file: main.bbl
\begin{thebibliography}{1}

\bibitem{GallowayLingSbierski}
{Gregory J. Galloway, Eric Ling, Jan Sbierski}.
\newblock {Timelike Completeness as an Obstruction to $C^0$-Extensions}.
\newblock {\em Communications in Mathematical Physics}, 359(3):937--949, 2018.

\bibitem{hirsch2012differential}
Morris~W. Hirsch.
\newblock {\em Differential Topology}.
\newblock Graduate Texts in Mathematics. Springer New York, 2012.

\bibitem{lee2003introduction}
John~M. Lee.
\newblock {\em Introduction to Smooth Manifolds}.
\newblock Graduate Texts in Mathematics. Springer, 2003.

\bibitem{Misner1973}
Charles~W. Misner, K.~S. Thorne, and J.~A. Wheeler.
\newblock {\em {Gravitation}}.
\newblock W. H. Freeman, San Francisco, 1973.

\bibitem{O'Neill}
Barrett O'Neill.
\newblock {\em Semi-Riemannian Geometry With Applications to Relativity}.
\newblock Academic Press, 1983.

\bibitem{continuousmetrics}
{Piotr T. Chruściel, James D.E. Grant}.
\newblock On lorentzian causality with continuous metrics.
\newblock {\em Class. Quantum Grav. 29}, 2012.

\bibitem{Schwarzschild2018}
Jan Sbierski.
\newblock {On the proof of the $C^0$-inextendibility of the Schwarzschild spacetime}.
\newblock {\em Journal of Differential Geometry}, 968(1), 2018.

\bibitem{Schwarzschild2016}
Jan Sbierski.
\newblock {The $C^0$-inextendibility of the Schwarzschild spacetime and the spacelike diameter in Lorentzian geometry}.
\newblock {\em Journal of Differential Geometry}, 108(2):319--378, 2018.

\end{thebibliography}
